\documentclass[graybox]{svmult}

\usepackage{mathptmx}       
\usepackage{helvet}         
\usepackage{courier}        
\usepackage{type1cm}        
\usepackage{makeidx}         
\usepackage{graphicx}        
\usepackage{multicol}        
\usepackage[bottom]{footmisc}

\usepackage{epsfig}
\usepackage{subfigure}
\usepackage{color}

  \usepackage{amsmath,amssymb,color}
  \usepackage{latexsym}

  \newcommand{\cX}{\mathcal{X}}
  \newcommand{\cY}{\mathcal{Y}}
  \newcommand{\cZ}{\mathcal{Z}}
  \newcommand{\R}{\mathbb{R}}
  \newcommand{\trans}{^T}
  \newcommand{\V}{\mathbf{V}}
  \newcommand{\0}{0}
  
  \renewcommand{\b}{b}
  \renewcommand{\c}{c}
  \newcommand{\e}{e}
  \newcommand{\f}{f}
  \newcommand{\g}{g}
  \renewcommand{\u}{u}
  \renewcommand{\v}{v}
  \newcommand{\w}{w}
  \newcommand{\x}{x}
  \newcommand{\y}{y}
  \newcommand{\z}{z}
  \newcommand{\rank}{\mathrm{rank\;}}
  
  \newcommand{\tr}{\mathop{\mathrm{tr}}\nolimits}
  \newcommand{\cE}{\mathcal{E}}
  \newcommand{\cS}{\mathcal{S}}

\makeindex             

\begin{document}

\title*{Two algorithms for compressed\\sensing of sparse tensors}
\author{Shmuel Friedland, Qun Li, Dan Schonfeld and Edgar A. Bernal}
\institute{Shmuel Friedland \at Department of Mathematics,
Statistics and Computer Science,
 University of Illinois at Chicago,
 Chicago, Illinois 60607-7045, USA. This work was supported by NSF grant
 DMS-1216393.
\email{friedlan@uic.edu} \and Qun Li \at PARC, Xerox Corporation,
800 Phillips Road, Webster, New York 14580, USA.
\email{Qun.Li@xerox.com} \and Dan Schonfeld \at Department of
Electrical and Computer Engineering, University of Illinois at
Chicago, Chicago, Illinois 60607, USA. \email{dans@uic.edu} \and
Edgar A. Bernal \at PARC, Xerox Corporation, 800 Phillips Road,
Webster, New York 14580, USA. \email{Edgar.Bernal@xerox.com}}

\maketitle

\abstract{Compressed sensing (CS) exploits the sparsity of a
signal in order to integrate acquisition and compression. CS
theory enables exact reconstruction of a sparse signal from
relatively few linear measurements via a suitable nonlinear
minimization process. Conventional CS theory relies on vectorial
data representation, which results in good compression ratios at
the expense of increased computational complexity. In applications
involving color images, video sequences, and multi-sensor
networks, the data is intrinsically of high-order, and thus more
suitably represented in tensorial form. Standard applications of
CS to higher-order data typically involve representation of the
data as long vectors that are in turn measured using large
sampling matrices, thus imposing a huge computational and memory
burden. In this chapter, we introduce Generalized Tensor
Compressed Sensing (GTCS)--a unified framework for compressed
sensing of higher-order tensors which preserves the intrinsic
structure of tensorial data with reduced computational complexity
at reconstruction. We demonstrate that GTCS offers an efficient
means for representation of multidimensional data by providing
simultaneous acquisition and compression from all tensor modes. In
addition, we propound two reconstruction procedures, a serial
method (GTCS-S) and a parallelizable method (GTCS-P), both capable
of recovering a tensor based on noiseless and noisy observations.
We then compare the performance of the proposed methods with
Kronecker compressed sensing (KCS) and multi-way compressed
sensing (MWCS). We demonstrate experimentally that GTCS
outperforms KCS and MWCS in terms of both reconstruction accuracy
(within a range of compression ratios) and processing speed. The
major disadvantage of our methods (and of MWCS as well), is that
the achieved compression ratios may be worse than those offered by
KCS.}

\section{Introduction}
Compressed sensing \cite{CS1, CS2} is a framework for
reconstructing signals that have sparse representations. A vector
$\x\in\R^N$ is called $k$-\emph{sparse} if $\x$ has at most $k$
nonzero entries. The sampling scheme can be modelled by a linear
operation. Assuming the number of measurements $m$ satisfies
$m<N$, and $A\in\R^{m\times N}$ is the matrix used for sampling,
then the encoded information is $\y\in \R^m$, where $\y=A\x$. The
decoder knows $A$ and recovers $\y$ by finding a solution $\hat
\z\in\R^N$ satisfying
\begin{equation}\label{l1minrec}
\hat{\z}=\arg \min_\z\|\z\|_1\quad \text{s.t.}\quad \y = A\z.
\end{equation}
Since $\|\cdot\|$ is a convex function and the set of all $\z$
satisfying $\y = A\z$ is convex, minimizing Eq.~\eqref{l1minrec}
is polynomial in $N$. Each $k$-sparse solution can be recovered
uniquely if $A$ satisfies the null space property (NSP) of order
$k$, denoted as NSP$_k$ \cite{NSP}. Given $A\in\R^{m\times N}$
which satisfies the NSP$_k$ property, a $k$-sparse signal
$\x\in\R^N$ and samples $\y=A\x$, recovery of $\x$ from $\y$ is
achieved by finding the $\z$ that minimizes Eq.~\eqref{l1minrec}.
One way to generate such $A$ is by sampling its entries using
numbers generated from a Gaussian or a Bernoulli distribution.
This matrix generation process guarantees that there exists a
universal constant $c$ such that if
\begin{equation}\label{nspcond}
m\ge 2ck\ln \frac{N}{k},
\end{equation}
then the recovery of $x$ using Eq.~\eqref{l1minrec} is successful
with probability greater than $1 - \exp(-\frac{m}{2c})$
\cite{Rauhut_compressivesensing}.

The objective of this document is to consider the case where the
$k$-sparse vector $\x$ is represented as a $k$-sparse tensor
$\cX=[x_{i_1,i_2,\ldots,i_d}]\in \R^{N_1\times N_2\times
\ldots\times N_d}$. Specifically, in the sampling phase, we
construct a set of measurement matrices $\{U_1,U_2,\ldots,U_d\}$
for all tensor modes, where $U_i\in \R^{m_i\times N_i}$ for
$i=1,2, \ldots,d$, and sample $\cX$ to obtain $\cY=\cX\times_1
U_1\times_2 U_2\times \ldots \times_d U_d\in\R^{m_1\times
m_2\times \ldots\times m_d}$ (see Sec.~\ref{multilinearAlgebra}
for a detailed description of tensor mode product notation). Note
that our sampling method is mathematically equivalent to that
proposed in \cite{KCS}, where $A$ is expressed as a Kronecker
product $A:=U_1\otimes U_2\otimes\ldots\otimes U_d$, which
requires $m$ to satisfy
\begin{equation}\label{nspcondmd}
m\ge 2ck(-\ln k+\sum_{i=1}^d\ln N_i).
\end{equation} We show that if each
$U_i$ satisfies the NSP$_k$ property, then we can recover $\cX$
uniquely from $\cY$ by solving a sequence of $\ell_1$ minimization
problems, each similar to the expression in Eq.~\eqref{l1minrec}.
This approach is advantageous relative to vectorization-based
compressed sensing methods such as that from \cite{KCS} because
the corresponding recovery problems are in terms of $U_i$'s
instead of $A$, which results in greatly reduced complexity. If
the entries of $U_i$ are sampled from Gaussian or Bernoulli
distributions, the following set of conditions needs to be
satisfied:
\begin{equation}\label{nspcondmdi}
m_i\ge 2ck\ln \frac{N_i}{k}, \quad i=1,\ldots,d.
\end{equation}
Observe that the dimensionality of the original signal $\cX$,
namely $N=N_1\cdot\ldots\cdot N_d$, is compressed to
$m=m_1\cdot\ldots\cdot m_d$. Hence, the number of measurements
required by our method must satisfy
\begin{equation}\label{totcompineq}
m\ge (2ck)^d \prod_{i=1}^d  \ln \frac{N_i}{k},
\end{equation}
which indicates a worse compression ratio than that from
Eq.~\eqref{nspcondmd}. This is consistent with the observations
from \cite{FLS13} (see Fig. 4(a) in \cite{FLS13}). We first
discuss our method for matrices, i.e., $d=2$, and then for
tensors, i.e., $d\ge 3$.

\section{Compressed Sensing of Matrices}
\subsection{Vector and Matrix Notation}
Column vectors are denoted by italic letters as
$\x=(x_1,\ldots,x_N)\trans\in\R^N$. Norms used for vectors include
\[\|\x\|_2:=\sqrt{\sum_{i=1}^N x_i^2}, \quad \|\x\|_1:=\sum_{i=1}^N |x_i|. \]

Let $[N]$ denote the set $\{1, 2, \ldots,N\}$, where $N$ is a
positive integer. Let $S\subset [N]$. We use the following
notation: $|S|$ is the cardinality of set $S$, $S^c:=[N]\setminus
S$, and $\|\x_S\|_1:=\sum_{i\in S}|x_i|$.

Matrices are denoted by capital italic letters as
$A=[a_{ij}]\in\R^{m\times N}$. The transposes of $\x$ and $A$ are
denoted by $\x\trans$ and $A\trans$ respectively. Norms of
matrices used include the Frobenius norm
$\|A\|_F:=\sqrt{\mathrm{tr}\; (AA\trans)}$, and the spectral norm
$\|A\|_2:=\max_{\|\x\|_2=1} \|A\x\|_2$. Let $R(X)$ denote the
column space of $X$. The singular value decomposition (SVD)
\cite{SVD} of $A$ with $\rank (A) = r$ is:
\begin{equation}\label{SVDA}
A=\sum_{i=1}^r (\sqrt{\sigma_i}\u_i)(\sqrt{\sigma_i}\v_i)\trans, \quad  \u_i\trans\u_j=\v_i\trans \v_j=\delta_{ij},\;i,j\in[r].
\end{equation}
Here, $\sigma_1(A)=\sigma_1\ge \ldots \ge \sigma_r(A)=\sigma_r>0$
are all positive singular values of $A$. $\u_i$ and $\v_i$ are the
left and the right singular vectors of $A$ corresponding to
$\sigma_i$.  Recall that
\[A\v_i=\sigma_i \u_i, \;A\trans \u_i=\sigma_i\v_i, \quad i\in[r], \quad \|A\|_2=\sigma_1(A), \quad \|A\|_F=\sqrt{\sum_{i=1}^r \sigma_i^2(A)}.\]
For $k< r$, let
$$A_k:=\sum_{i=1}^k (\sqrt{\sigma_i}\u_i)(\sqrt{\sigma_i}\v_i)\trans.$$
For $k\ge r$,  we have $A_k:=A$. Then $A_k$ is a solution to the
following minimization problems:
\begin{eqnarray*}
&&\min_{B\in\R^{m\times N}, \rank (B)\le k} \|A-B\|_F=\|A-A_k\|_F=\sqrt{\sum_{i=k+1}^r \sigma_i^2(A)},\\
&&\min_{B\in\R^{m\times N}, \rank (B)\le k}
\|A-B\|_2=\|A-A_k\|_2=\sigma_{k+1}(A).
\end{eqnarray*}
We call $A_k$ the best rank-$k$ approximation to $A$. Note that
$A_k$ is unique if and only if $\sigma_j(A)>\sigma_{j+1}(A)$ for
$j\in [k-1]$.

$A\in \R^{m\times N}$ satisfies the \emph{null space property of
order} $k$, abbreviated as NSP$_k$ property, if the following
condition holds: let $A\w=\0, \w\ne \0$; then for each $S\subset
[N]$ satisfying $|S|=k$, the inequality
$\|\w_S\|_1<\|\w_{S^c}\|_1$ is satisfied.

Let $\Sigma_{k,N}\subset \R^N$ denote all vectors in $\R^N$ which
have at most $k$ nonzero entries. The fundamental lemma of
noiseless recovery in compressed sensing that has been introduced
in Chapter 1 is:
\begin{lemma}\label{fundlemnrec}
Suppose that $A\in\R^{m\times N}$ satisfies the NSP$_k$ property.
Assume that $\x\in\Sigma_{k,N}$ and let $\y=A\x$.  Then for each
$\z\in\R^N$ satisfying $A\z=\y$, $\|\z\|_1\ge \|\x\|_1$.  Equality
holds if and only if~$\z=\x$. That is, $\x= \arg \min_{\z}
\|\z\|_1 \quad \text{s.t.}\quad \y = A\z.$  The complexity of this
minimization problem is $O(N^3)$ \cite{Complexity, Fast}.
\end{lemma}
\subsection{Noiseless Recovery}
\subsubsection{Compressed Sensing of Matrices - Serial
Recovery (CSM-S)} \label{sec:2} The serial recovery method for
compressed sensing of matrices in the noiseless case is described
by the following theorem.
\begin{theorem}[CSM-S]\label{CSM-S}
\label{compsensmatS} Let $X=[x_{ij}]\in\R^{N_1\times N_2}$ be
$k$-sparse.  Let $U_i\in \R^{m_i\times N_i}$ and assume that $U_i$
satisfies the NSP$_k$ property for $i\in [2]$.  Define
\begin{equation}\label{defmatY}
Y=[y_{pq}]=U_1 X U_2\trans \in \R^{m_1\times m_2}.
\end{equation}
Then $X$ can be recovered uniquely as follows. Let
$\y_1,\ldots,\y_{m_2}\in\R^{m_1}$ be the columns of $Y$. Let $\hat
\z_i\in\R^{N_1}$ be a solution of
\begin{equation}\label{defzstarMatrix}
\hat \z_i=\arg\min_{\z_i}\|\z_i\|_1 \quad \text{s.t.}\quad \y_i =
U_1\z_i, \quad i\in [m_2].
\end{equation}
Then each $\hat \z_i$ is unique and $k$-sparse.  Let
$Z\in\R^{N_1\times m_2}$ be the matrix whose columns are $\hat
\z_1,\ldots,\hat \z_{m_2}$. Let $\w_1\trans,\ldots,
\w_{N_1}\trans$ be the rows of $Z$.  Then $\v_j\in\R^{N_2}$, whose
transpose is the $j$-th row of $X$, is the solution of
\begin{equation}\label{defustarMatrix}
\hat \v_j=\arg\min_{\v_j}\|\v_j\|_1 \quad \text{s.t.}\quad \w_j=
U_2 \v_j, \quad j\in [N_1].
\end{equation}
\end{theorem}
\begin{proof}
Let $Z$ be the matrix whose columns are $\hat \z_1,\ldots,\hat
\z_{m_2}$.  Then $Z$ can be written as
$Z=XU_2\trans\in\R^{N_1\times m_2}$. Note that $\hat \z_i$ is a
linear combination of the columns of $X$. $\hat \z_i$ has at most
$k$ nonzero coordinates, because the total number of nonzero
elements in $X$ is $k$. Since $Y=U_1 Z$, it follows that
$\y_i=U_1\hat \z_i$. Also, since $U_1$ satisfies the NSP$_k$
property, we arrive at Eq.~\eqref{defzstarMatrix}. Observe that
$Z\trans=U_2 X\trans$; hence, $\w_j=U_2 \hat \v_j$. Since $X$ is
$k$-sparse, then each $\hat \v_j$ is $k$-sparse.  The assumption
that $U_2$ satisfies the NSP$_k$ property implies
Eq.~\eqref{defustarMatrix}. \qed
\end{proof}
If the entries of $U_1$ and $U_2$ are drawn from random
distributions as described above, then the set of conditions from
Eq.~\eqref{nspcondmdi} needs to be met as well. Note that although
Theorem \ref{CSM-S} requires both $U_1$ and $U_2$ to satisfy the
NSP$_k$ property, such constraints can be relaxed if each row of
$X$ is $k'$-sparse, where $k'<k$.  In this case, it follows from
the proof of Theorem \ref{CSM-S} that $X$ can be recovered as long
as $U_1$ and $U_2$ satisfy the NSP$_k$ and the NSP$_{k'}$
properties respectively.

\subsubsection{Compressed Sensing of Matrices - Parallelizable Recovery (CSM-P)}
The parallelizable recovery method for compressed sensing of
matrices in the noiseless case is described by the following
theorem.
\begin{theorem}[CSM-P]\label{CSM-P}
Let $X=[x_{ij}]\in\R^{N_1\times N_2}$ be $k$-sparse.  Let $U_i\in
\R^{m_i\times N_i}$ and assume that $U_i$ satisfies the NSP$_k$
property for $i \in [2]$.  If Y is given by Eq.~\eqref{defmatY},
then $X$ can be recovered approximately as follows. Consider a
rank decomposition (e.g., SVD) of $Y$ such that
\begin{equation}\label{ranklikedec1Y}
Y=\sum_{i=1}^K \b_i^{(1)} (\b_i^{(2)})\trans,
\end{equation}
where $K = \rank(Y)$. Let $\hat \w_i^{(j)}\in \mathbb{R}^{N_j}$ be
a solution of
\begin{equation}\label{defzstarP}
\hat \w_i^{(j)}=\arg\min_{\w_i}\|\w_i^{(j)}\|_1 \quad
\text{s.t.}\quad \b_i^{(j)} = U_j\w_i^{(j)}, \quad i\in [K], j\in
[2].\nonumber
\end{equation}
Then each $\hat \w_i^{(j)}$ is unique and $k$-sparse, and
\begin{equation}\label{ranklikedecYP}
X=\sum_{i=1}^K \hat \w_i^{(1)} (\hat \w_i^{(2)})\trans.
\end{equation}
\end{theorem}

\begin{proof}  First observe that $R( Y)\subset U_1R( X)$ and $R( Y\trans)\subset U_2R( X\trans)$.
Since Eq.~\eqref{ranklikedec1Y} is a rank decomposition of $Y$, it
follows that $\b_i^{(1)}\in U_1R( X)$ and $\b_i^{(2)}\in U_2R(
X\trans)$. Hence $\hat \w_i^{(1)}\in R( X), \hat \w_i^{(2)}\in R(
X\trans)$ are unique and $k$-sparse.  Let $\hat X:=\sum_{i=1}^K
\hat \w_i^{(1)} (\hat \w_i^{(2)})\trans$. Assume to the contrary
that $X-\hat X\ne 0$.  Clearly $R(X-\hat X)\subset R (X), R
(X\trans - \hat X\trans)\subset R( X\trans)$. Let $X-\hat
X=\sum_{i=1}^{J} \u_i^{(1)} (\u_i^{(2)})\trans$ be a rank
decomposition of $X-\hat X$.  Hence
$\u_1^{(1)},\ldots,\u_J^{(1)}\in R( X)$ and
$\u_1^{(2)},\ldots,\u_J^{(2)}\in R( X\trans)$ are two sets of $J$
linearly independent vectors.   Since each vector either in $R(
X)$ or in $R( X\trans)$ is $k$-sparse, and $U_1, U_2$ satisfy the
NSP$_k$ property, it follows that $U_1\u_1^{(j)},\ldots,
U_1\u_J^{(j)}$ are linearly independent for $j\in [2]$ (see
Appendix for proof).  Hence the matrix
$Z:=\sum_{i=1}^J(U_1\u_i^{(1)}) (U_2\u_i^{(2)})\trans$ has rank
$J$. In particular, $Z\ne 0$.  On the other hand, $Z=U_1(X-\hat
X)U_2\trans=Y-Y=0$, which contradicts the previous statement.  So
$X=\hat X$. \qed
\end{proof}

The above recovery procedure consists of two stages, namely, the
decomposition stage and the reconstruction stage, where the latter
can be implemented in parallel for each matrix mode. Note that the
above theorem is equivalent to multi-way compressed sensing for
matrices (MWCS) introduced in \cite{MWCS}.

\subsubsection{Simulation Results}\label{NoiselessImageSimu}
We demonstrate experimentally the performance of GTCS methods on
the reconstruction of sparse images and video sequences. As
demonstrated in \cite{KCS}, KCS outperforms several other methods
including independent measurements and partitioned measurements in
terms of reconstruction accuracy in tasks related to compression
of multidimensional signals. A more recently proposed method is
MWCS, which stands out for its reconstruction efficiency. For the
above reasons, we compare our methods with both KCS and MWCS. Our
experiments use the $\ell_1$-minimization solvers from \cite{l1}.
We set the same threshold to determine the termination of the
$\ell_1$-minimization process in all subsequent experiments. All
simulations are executed on a desktop with a 2.4 GHz Intel Core i7
CPU and 16GB RAM.

The original grayscale image (see Fig. \ref{original}) is of size
$128\times 128$ pixels ($N=16384$). We use the discrete cosine
transform (DCT) as the sparsifying transform, and zero-out the
coefficients outside the $16\times 16$ sub-matrix in the upper
left corner of the transformed image. We refer to the inverse DCT
of the resulting sparse set of transform coefficients as the
target image. Let $m$ denote the number of measurements along both
matrix modes; we generate the measurement matrices with entries
drawn from a Gaussian distribution with mean 0 and standard
deviation $\sqrt{\frac{1}{m}}$. For simplicity, we set the number
of measurements for two modes to be equal; that is, the randomly
constructed Gaussian matrix $U$ is of size $m\times 128$ for each
mode. Therefore, the KCS measurement matrix $U\otimes U$ is of
size $m^2\times 16384$, and the total number of measurements is
$m^2$. We refer to $\frac{m^2}{N}$ as the normalized number of
measurements. For GTCS, both the serial recovery method GTCS-S and
the parallelizable recovery method GTCS-P are implemented. In the
matrix case, for a given choice of rank decomposition method,
GTCS-P and MWCS are equivalent; in this case, we use SVD as the
rank decomposition approach. Although the reconstruction stage of
GTCS-P is parallelizable, we recover each vector in series.
Consequently, we note that the reported performance data for
GTCS-P can be improved upon. We examine the performance of the
above methods by varying the normalized number of measurements
from 0.1 to 0.6 in steps of 0.1. Reconstruction performance for
the different methods is compared in terms of reconstruction
accuracy and computational complexity. Reconstruction accuracy is
measured via the peak signal to noise ratio (PSNR) between the
recovered and the target image (both in the spatial domain),
whereas computational complexity is measured in terms of the
reconstruction time (see Fig. \ref{image_noiseless}).

\begin{figure}[htb]
\begin{center}
\includegraphics[width=0.4\linewidth]{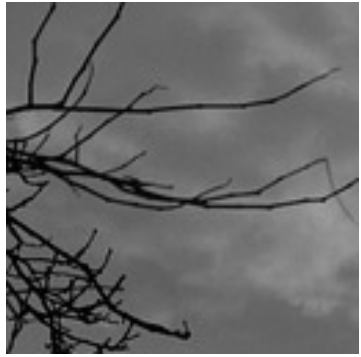}
\caption{The original grayscale image.} \label{original}
\end{center}
\end{figure}

\begin{figure}[htb]
\begin{center}
\subfigure[PSNR
comparison]{\label{image_noiseless_PSNR}\includegraphics[width=0.49\linewidth]{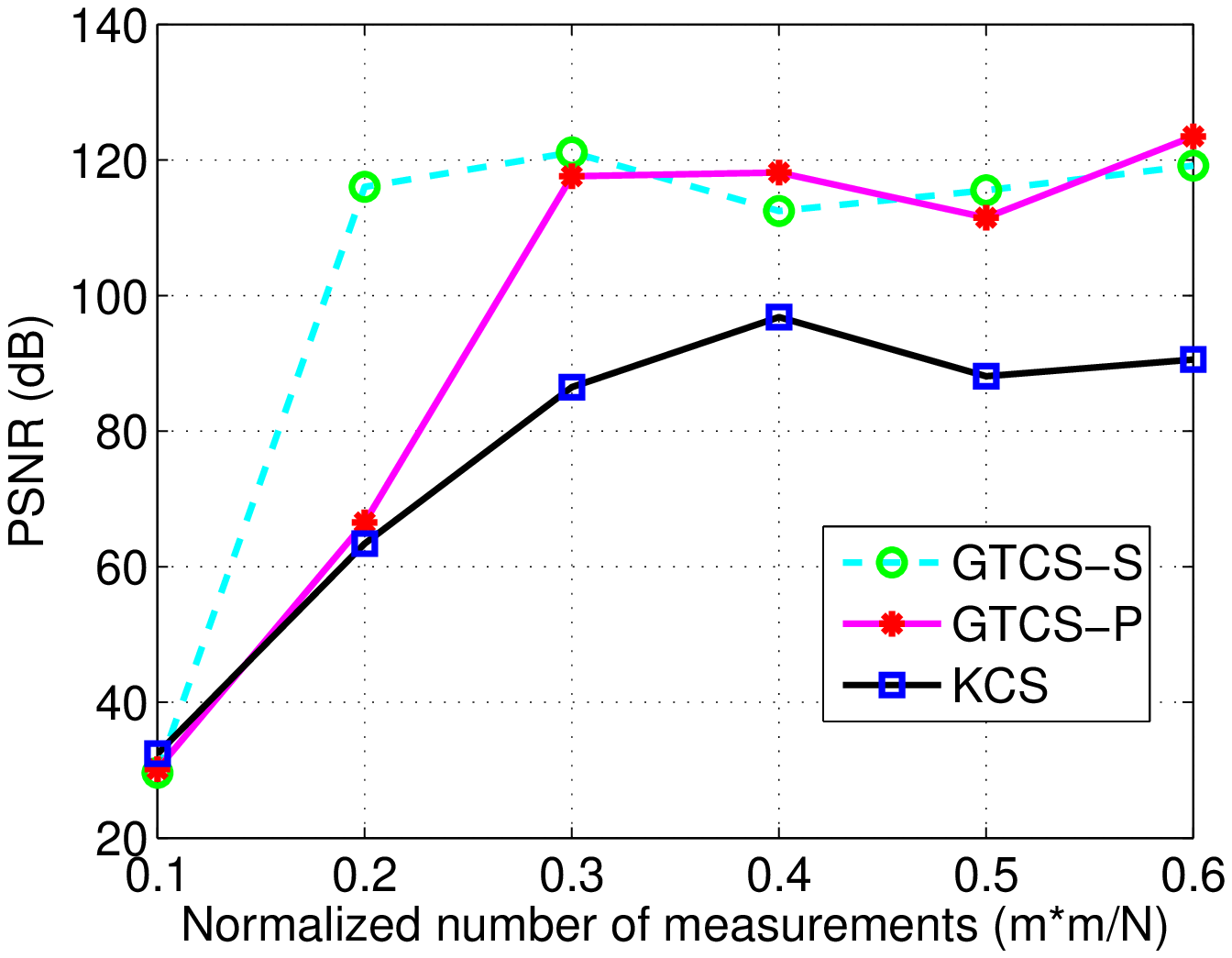}}
\subfigure[Recovery time
comparison]{\label{UICTime}\includegraphics[width=0.49\linewidth]{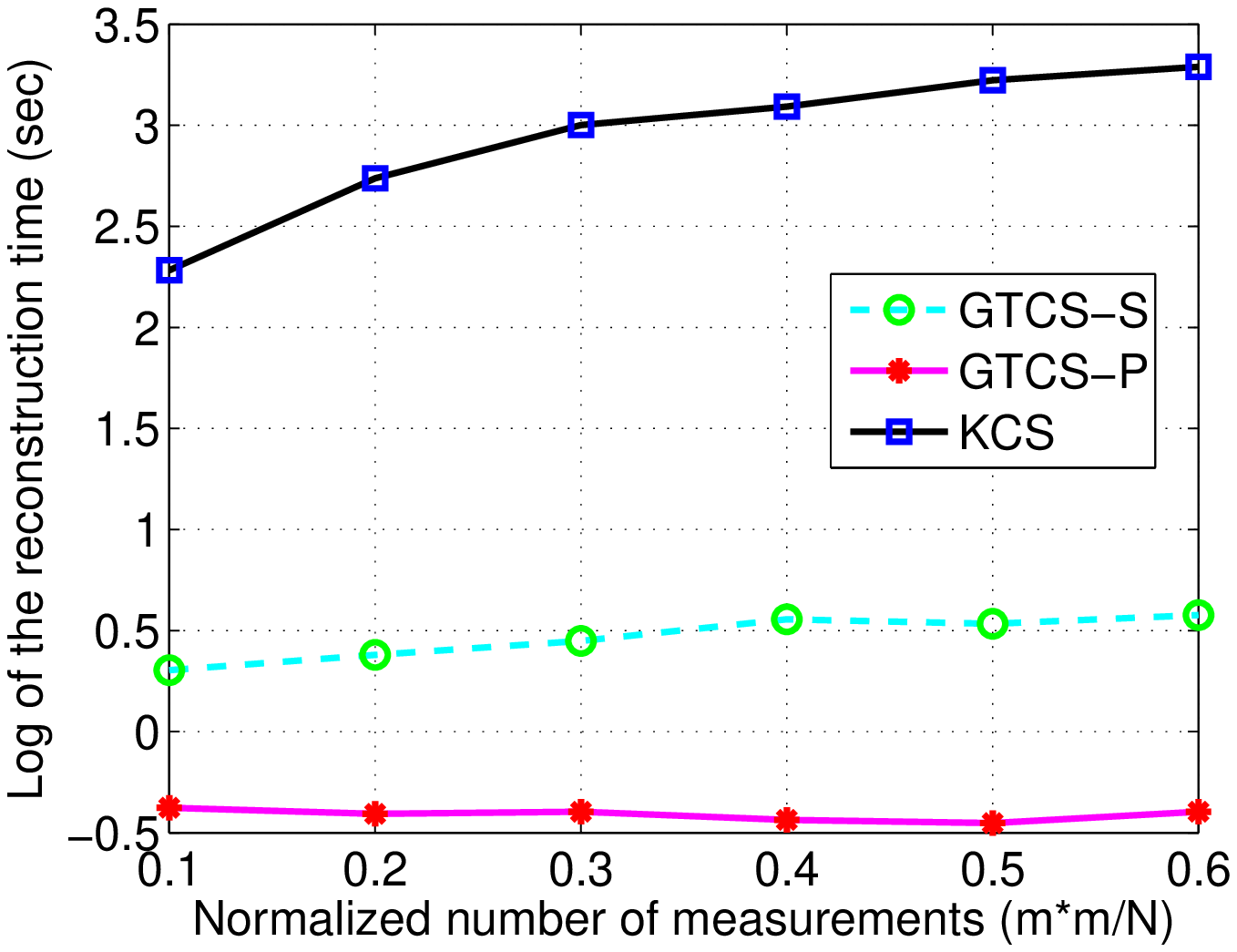}}
\caption{Performance comparison among the tested methods in terms
of PSNR and reconstruction time in the scenario of noiseless
recovery of a sparse image.} \label{image_noiseless}
\end{center}
\end{figure}

\subsection{Recovery of Data in the Presence of Noise}
Consider the case where the observation is noisy. For a given
integer $k$, a matrix $A\in \R^{m\times N}$ satisfies the
restricted isometry property (RIP$_k$) \cite{RIP} if
 $$(1-\delta_k)\|\x\|^2_2\le \|A\x\|^2_2\le (1+\delta_k)\|\x\|^2_2$$
 for all $\x\in\Sigma_{k,N}$ and for some $\delta_k\in(0,1)$.

It was shown in \cite{C2} that the reconstruction in the presence
of noise is achieved by solving
\begin{equation}\label{noisyrecoveryEqu}
\hat \x=\arg\min_{\z} \|\z\|_1, \quad \text{s.t.} \quad
\|A\z-\y\|_2 \le \varepsilon,
\end{equation}

which has complexity $O(N^3)$.

\begin{lemma}\label{sigrecer}  Assume that $A\in \R^{m\times N}$ satisfies the RIP$_{2k}$ property for some $\delta_{2k}\in (0, \sqrt{2}-1)$.
Let $\x\in\Sigma_{k,N},  \y=A\x+\e,$ where $e$ denotes the noise
vector, and $ \|\e\|_2\le \varepsilon$ for some real nonnegative
number $\varepsilon$. Then
\begin{equation}\label{sigrecer1}
 \|\hat \x -\x\|_2\le C_2\varepsilon,  \textrm{ where } C_2=\frac{4\sqrt{1+\delta_{2k}}}{1-(1+\sqrt{2})\delta_{2k}}.
\end{equation}
\end{lemma}

\subsubsection{Compressed Sensing of Matrices - Serial Recovery (CSM-S) in the Presence of Noise}
The serial recovery method for compressed sensing of matrices in
the presence of noise is described by the following theorem.
\begin{theorem}[CSM-S in the presence of noise]
\label{compsensmatSN} Let $X=[x_{ij}]\in\R^{N_1\times N_2}$ be
$k$-sparse.  Let $U_i\in \R^{m_i\times N_i}$ and assume that $U_i$
satisfies the RIP$_{2k}$ property for some $\delta_{2k}\in (0,
\sqrt{2}-1)$, $i\in [2]$.  Define
\begin{equation}\label{defmatYE}
Y=[y_{pq}]=U_1 X U_2\trans +E, \quad Y \in \R^{m_1\times m_2},
\end{equation}
where $E$ denotes the noise matrix, and $\|E\|_F\le \epsilon$ for
some real nonnegative number $\epsilon$. Then $X$ can be recovered
approximately as follows. Let
$\c_1(Y),\ldots,\c_{m_2}(Y)\in\R^{m_1}$ denote the columns of $Y$.
Let $\hat \z_i\in\R^{N_1}$ be a solution of
\begin{equation}\label{defzstar}
\hat \z_i=\arg\min_{\z_i}\|\z_i\|_1 \quad \text{s.t.}\quad
\|\c_i(Y) - U_1\z_i\|_2\le \epsilon, \quad i\in [m_2].
\end{equation}
Let $Z\in\R^{N_1\times m_2}$ be the matrix whose columns are $\hat
\z_1,\ldots,\hat \z_{m_2}$. According to Eq.~\eqref{sigrecer1},
$\|\c_i(Z)-\c_i(XU_2\trans)\|_2=\|\hat
\z_i-\c_i(XU_2\trans)\|_2\le C_2\varepsilon$, hence
$\|Z-XU_2\trans\|_F\le \sqrt{m_2}C_2\varepsilon$. Let
$\c_1(Z\trans),\ldots, \c_{N_1}(Z\trans)$ be the rows of $Z$. Then
$\u_j\in\R^{N_2}$, the $j$-th row of $X$, is the solution of
\begin{equation}\label{defustar}
\hat \u_j=\arg\min_{\u_j}\|\u_j\|_1 \quad \text{s.t.}\quad
\|\c_j(Z\trans)-U_2 u_j\|_2 \le \sqrt{m_2}C_2\varepsilon, \quad
j\in [N_1].
\end{equation}
Denote by $\hat X$ the recovered matrix, then according to
Eq.~\eqref{sigrecer1},
\begin{equation}\label{MatrixBound}
\|\hat X- X\|_F\le \sqrt{m_2N_1}C_2^2\varepsilon.
\end{equation}
\end{theorem}
\begin{proof} The proof of the theorem follows from Lemma \ref{sigrecer}.
\qed
\end{proof}

The upper bound in Eq. \eqref{MatrixBound} can be tightened by
assuming that the entries of $E$ adhere to a specific type of
distribution. Let $E=[\e_1,\ldots,\e_{m_2}]$. Suppose that each
entry of $E$ is an independent random variable with a given
distribution having zero mean.  Then we can assume that
$\|\e_j\|_2\le\frac{\varepsilon}{\sqrt{m_2}}$, which implies that
$\|E\|_F\le \varepsilon$.

Each $\z_i$ can be recovered by finding a solution to
\begin{equation}\label{hatzirecov}
\hat \z_i=\arg\min_{\z_i}\|\z_i\|_1 \quad \text{s.t.}\quad
\|\c_i(Y) - U_1\z_i\|_2\le \frac{\varepsilon}{\sqrt{m_2}}, \quad
i\in [m_2].
\end{equation}

Let $Z=[\hat \z_1 \ldots \hat \z_{m_2}]\in\R^{N_1\times m_2}$.
According to Eq.~\eqref{sigrecer1},
$\|\c_i(Z)-\c_i(XU_2\trans)\|_2=\|\hat
\z_i-\c_i(XU_2\trans)\|_2\le C_2\frac{\varepsilon}{\sqrt{m_2}}$;
therefore $\|Z-XU_2\trans\|_F\le C_2\varepsilon$.

Let $E_1:= Z-XU_2\trans$ be the error matrix, and assume that the
entries of $E_1$ adhere to the same distribution as the entries of
$E$. Hence, $\|\c_i(Z\trans)-\c_i(U_2X\trans)\|_2\le
\frac{C_2\varepsilon}{\sqrt{N_1}}$.

$\hat X$ can be reconstructed by recovering each row of $X$:
\begin{equation}\label{hatbirecov}
\hat \u_j=\arg\min_{\u_j}\|\u_j\|_1 \quad \text{s.t.}\quad
\|\c_j(Z\trans)-U_2 \u_j\|_2 \le
\frac{C_2\varepsilon}{\sqrt{N_1}}, \quad j\in [N_1].
\end{equation}

Consequently, $\|\hat \u_j- \c_j(X\trans)\|_2\le
\frac{C_2^2\varepsilon}{\sqrt{N_1}}$, and the recovery error is
bounded as follows:
\begin{equation}\label{improvmaterestim}
\|\hat X- X\|_F\le C_2^2\varepsilon.
\end{equation}

When $Y$ is not full-rank, the above procedure is equivalent to
the following alternative. Let $Y_k$ be a best rank-$k$
approximation of $Y$:
\begin{equation}\label{SVDkaproxY}
Y_k=\sum_{i=1}^k (\sqrt{\tilde\sigma_i}\tilde \u_i)(\sqrt{\tilde
\sigma_i}\tilde \v_i)\trans.
\end{equation}

Here, $\tilde \sigma_i$ is the $i$-th singular value  of $Y$, and
$\tilde \u_i,\tilde \v_i$ are the corresponding left and right
singular vectors of $Y$ for $i\in [k]$, assume that $k\leq
\min(m_1, m_2)$. Since $X$ is assumed to be $k$-sparse, then
$\rank(X)\le k$. Hence the ranks of $XU_2$ and $U_1XU_2\trans$ are
less than or equal to $k$. In this case, recovering $X$ amounts to
following the procedure described above with $Y_k$ and $Z_k$
taking the place of $Y$ and $Z$ respectively.

\subsubsection{Compressed Sensing of Matrices - Parallelizable Recovery (CSM-P) in the Presence of Noise}
The parallelizable recovery method for compressed sensing of
matrices in the presence of noise is described by the following
theorem.
\begin{theorem}[CSM-P in the presence of noise]
\label{compsensmatPN} Let $X=[x_{ij}]\in\R^{N_1\times N_2}$ be
$k$-sparse.  Let $U_i\in \R^{m_i\times N_i}$ and assume that $U_i$
satisfies the RIP$_{2k}$ property for some $\delta_{2k}\in (0,
\sqrt{2}-1)$, $i\in [2]$.  Let $Y$ be as defined in
Eq.~\eqref{defmatYE}. Then $X$ can be recovered uniquely as
follows. Let $Y_{k'}$ be a best rank-$k'$ approximation of $Y$ as
in Eq.~\eqref{SVDkaproxY}, where $k'$ is the minimum of $k$ and
the number of singular values of $Y$ greater than
$\frac{\varepsilon}{\sqrt{k}}$. Then $\hat X = \sum_{i=1}^{k'}
\frac{1}{\tilde\sigma_i} \hat \x_i {\hat \y_i}\trans$ and
\begin{equation}\label{esterrMS}
\|X-\hat X\|_F\le C^2\varepsilon,
\end{equation}

where
\begin{eqnarray}\label{hatxrecov}
\hat \x_i=\arg\min_{\x_i} \|\x_i\|_1\quad \text{s.t.}\quad \|\tilde\sigma_i\tilde \u_i - U_1\x_i\|_2 \le \frac{\varepsilon}{\sqrt{2k}},&&\nonumber\\
\label{hatyrecov} \hat \y_i=\arg\min_{\y_i} \|\y_i\|_1\quad
\text{s.t.}\quad \|\tilde\sigma_i\tilde \v_i - U_2\y_i\|_2 \le
\frac{\varepsilon}{\sqrt{2k}},&&\\ i\in [k].&& \nonumber
\end{eqnarray}

\end{theorem}

\begin{proof}  Assume that $k<\min(m_1,m_2)$, otherwise $Y_k=Y$.
Since $\rank(U_1XU_2)\leq k$, $Y_k=U_1XU_2+E_k$. Let
\begin{equation}\label{SVDdec}
U_1XU_2\trans=\sum_{i=1}^k (\sqrt{\sigma_i}
\u_i)(\sqrt{\sigma_i}\v_i)\trans
\end{equation}
be the SVD of $U_1XU_2\trans$. Then $\|\u_i\|=\|\tilde
\u_i\|=\|\v_i\|=\|\tilde \v_i\|=1$ for $i\in [k]$.

Assuming
\begin{equation}\label{mainhyp}
\e_i:=\sqrt{\tilde\sigma_i}\tilde \u_i-\sqrt{\sigma_i} \u_i, \quad
\f_i:=\sqrt{\tilde\sigma_i}\tilde \v_i-\sqrt{\sigma_i} \v_i, \quad
i\in [k],
\end{equation}
then the entries of $ \e_i$ and $\f_i$ are independent Gaussian
variables with zero mean and standard deviation
$\frac{\varepsilon}{\sqrt{2\sigma_im_1k}}$ and
$\frac{\varepsilon}{\sqrt{2\sigma_im_2k}}$, respectively, for
$i\in[k]$. When $\varepsilon^2\ll \varepsilon$,
\begin{equation}\label{Esapprox}
E_k\approx \sum_{i=1}^k \e_i(\sqrt{\sigma_i}\v_i\trans)+\sum_{i=1}^k
(\sqrt{\sigma_i} \u_i)\f_i\trans.
\end{equation}

In this scenario,
\begin{equation}\label{eruvas}
\|\sqrt{\sigma_i}\u_i-\sqrt{\tilde\sigma_i}\tilde \u_i\|\le
\frac{\varepsilon}{\sqrt{2k\sigma_i}}, \quad
\|\sqrt{\sigma_i}\v_i-\sqrt{\tilde\sigma_i}\tilde \v_i\|\le
\frac{\varepsilon}{\sqrt{2k\sigma_i}}.
\end{equation}

Note that
\begin{equation}\label{distsingest}
\sum_{i=1}^{\min(m_1,m_2)}(\sigma_i-\sigma(Y_k))^2\le \tr (E
E\trans)\le\varepsilon^2, \quad \sum_{i=1}^k
(\sigma_i-\tilde\sigma_i)^2\le \tr (E_k
E_k\trans)\le\varepsilon^2.
\end{equation}

Given the way $k'$ is defined, it can be interpreted as the
numerical rank of $Y$. Consequently, $Y$ can be well represented by its best rank $k'$
approximation.  Thus
\begin{equation}\label{SVDdec'}
U_1XU_2\trans\approx\sum_{i=1}^{k'}
(\sqrt{\sigma_i}\u_i)(\sqrt{\sigma_i}\v_i\trans), \quad
Y_{k'}=\sum_{i=1}^{k'} (\sqrt{\tilde\sigma_i}\tilde
\u_i)(\sqrt{\tilde \sigma_i}\tilde \v_i\trans),\quad i\in [k'].
\end{equation}
Assuming $\sigma_i\approx \tilde\sigma_i$ for $i\in [k']$, we
conclude that
\begin{equation}\label{esteruv}
\|\tilde\sigma_i\tilde \u_i-\sigma_i \u_i\|\le
\frac{\varepsilon}{\sqrt{2k}}, \quad \|\tilde\sigma_i\tilde
\v_i-\sigma_i \v_i\| \le \frac{\varepsilon}{\sqrt{2k}}.
\end{equation}
A compressed sensing framework can be used to solve the following
set of minimization problems, for $i\in [k']$:
\begin{eqnarray}\label{hatxrecov}
\hat \x_i=\arg\min_{\x_i} \|\x_i\|_1\quad \text{s.t.}\quad \|\tilde\sigma_i\tilde \u_i - U_1\x_i\|_2 \le \frac{\varepsilon}{\sqrt{2k}},\\
\label{hatyrecov} \hat \y_i=\arg\min_{\y_i} \|\y_i\|_1\quad
\text{s.t.}\quad \|\tilde\sigma_i\tilde \v_i - U_2\y_i\|_2 \le
\frac{\varepsilon}{\sqrt{2k}}.
\end{eqnarray}
The error bound from Eq.~\eqref{esterrMS} follows. \qed
\end{proof}

\subsubsection{Simulation Results}
In this section, we use the same target image and experimental
settings used in Section \ref{NoiselessImageSimu}. We simulate the
noisy recovery scenario by modifying the observation with
additive, zero-mean Gaussian noise having standard deviation
values ranging from 1 to 10 in steps of 1, and attempt to recover
the target image using Eq. \eqref{noisyrecoveryEqu}.  As before,
reconstruction performance is measured in terms of PSNR between
the recovered and the target image, and in terms of reconstruction
time, as illustrated in Figs. \ref{NoisyImagePSNR} and
\ref{NoisyImageTime}.

\begin{figure}[htb]
\begin{center}
\subfigure[GTCS-S]{\label{UICPSNR}\includegraphics[width=0.32\linewidth]{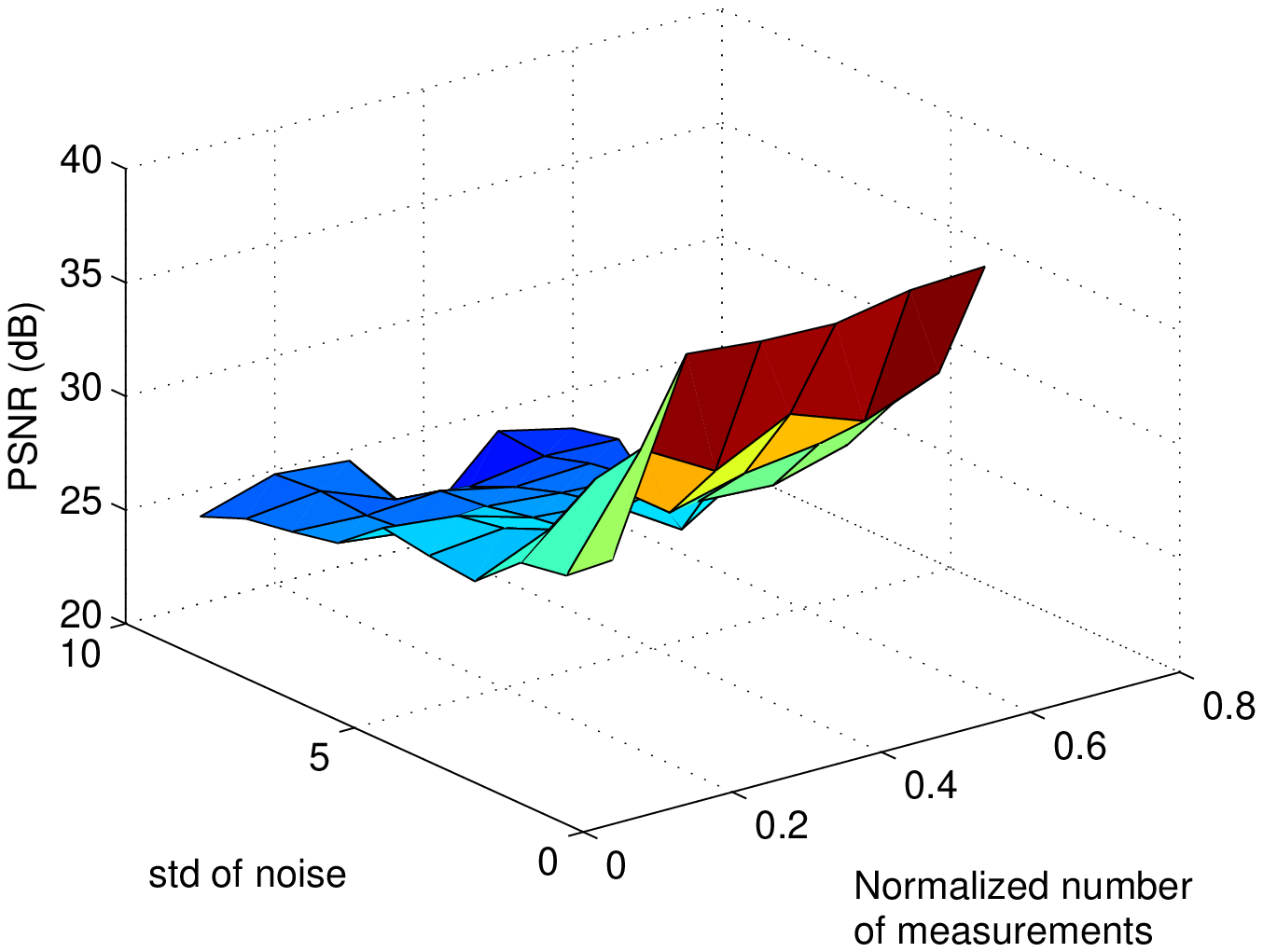}}
\subfigure[GTCS-P]{\label{UICTime}\includegraphics[width=0.32\linewidth]{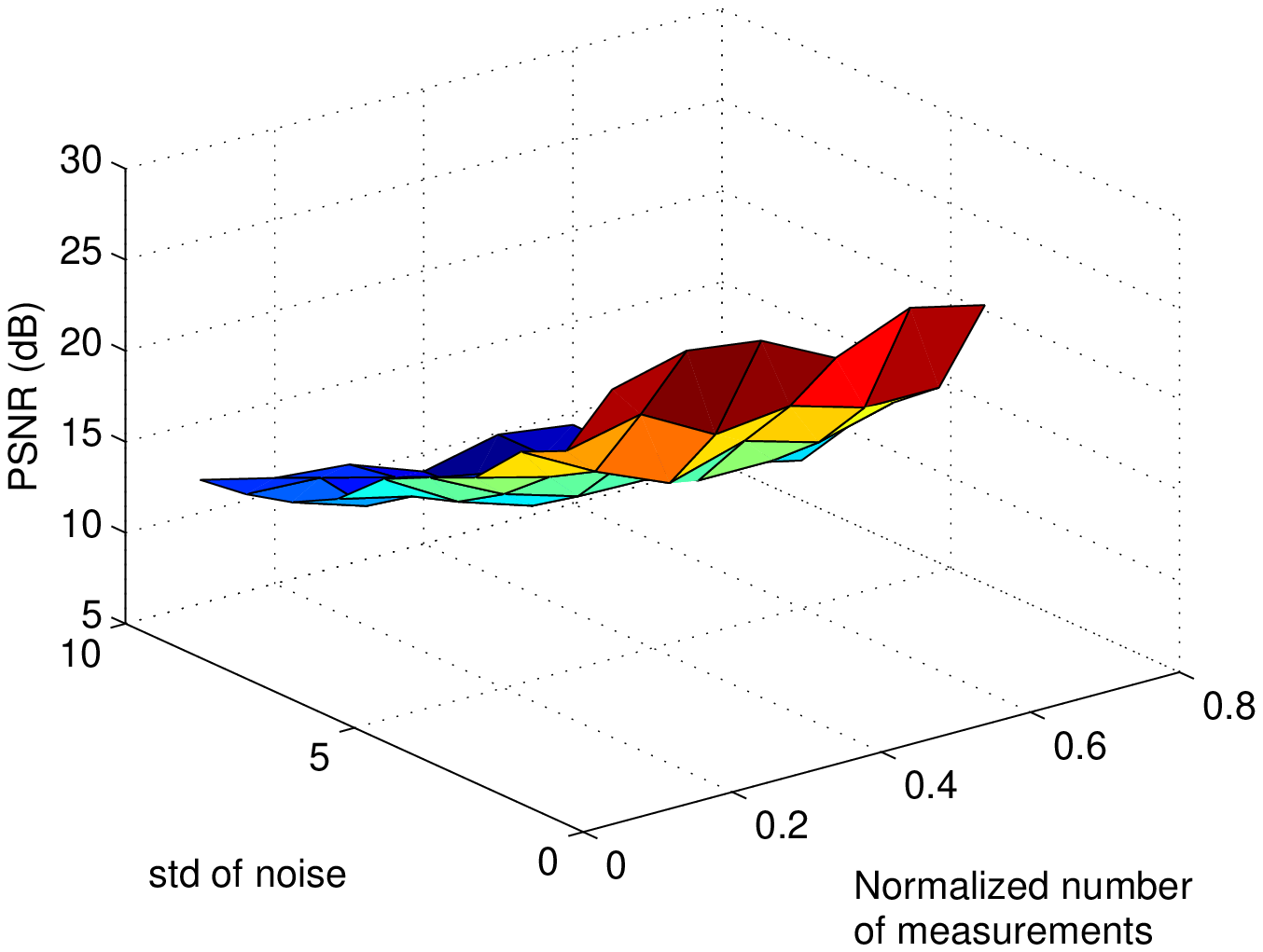}}
\subfigure[KCS]{\label{UICTime}\includegraphics[width=0.32\linewidth]{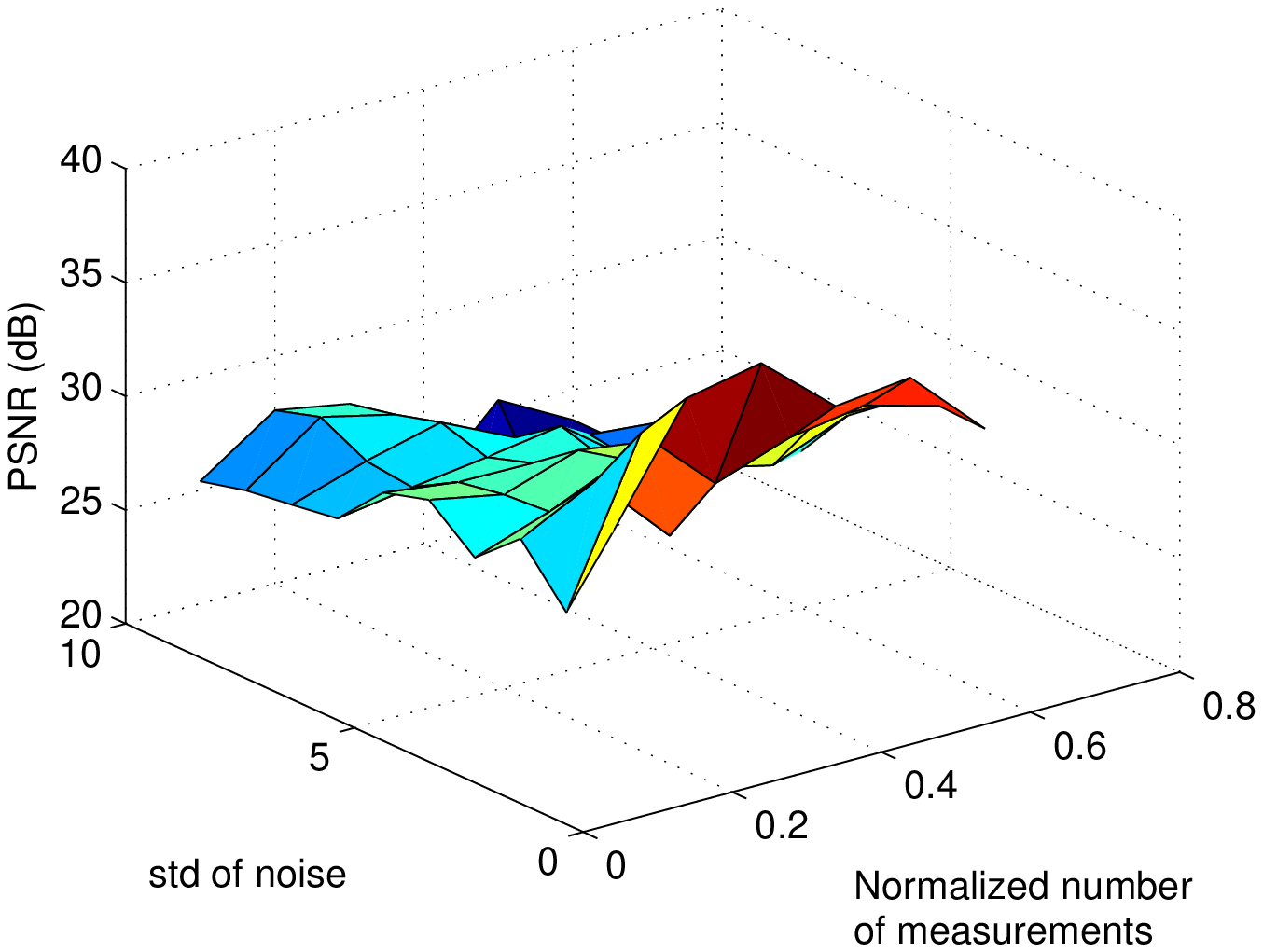}}
\caption{PSNR between target and recovered image for the tested methods in the noisy recovery scenario.} \label{NoisyImagePSNR}
\end{center}
\end{figure}

\begin{figure}[htb]
\begin{center}
\subfigure[GTCS-S]{\label{UICPSNR}\includegraphics[width=0.32\linewidth]{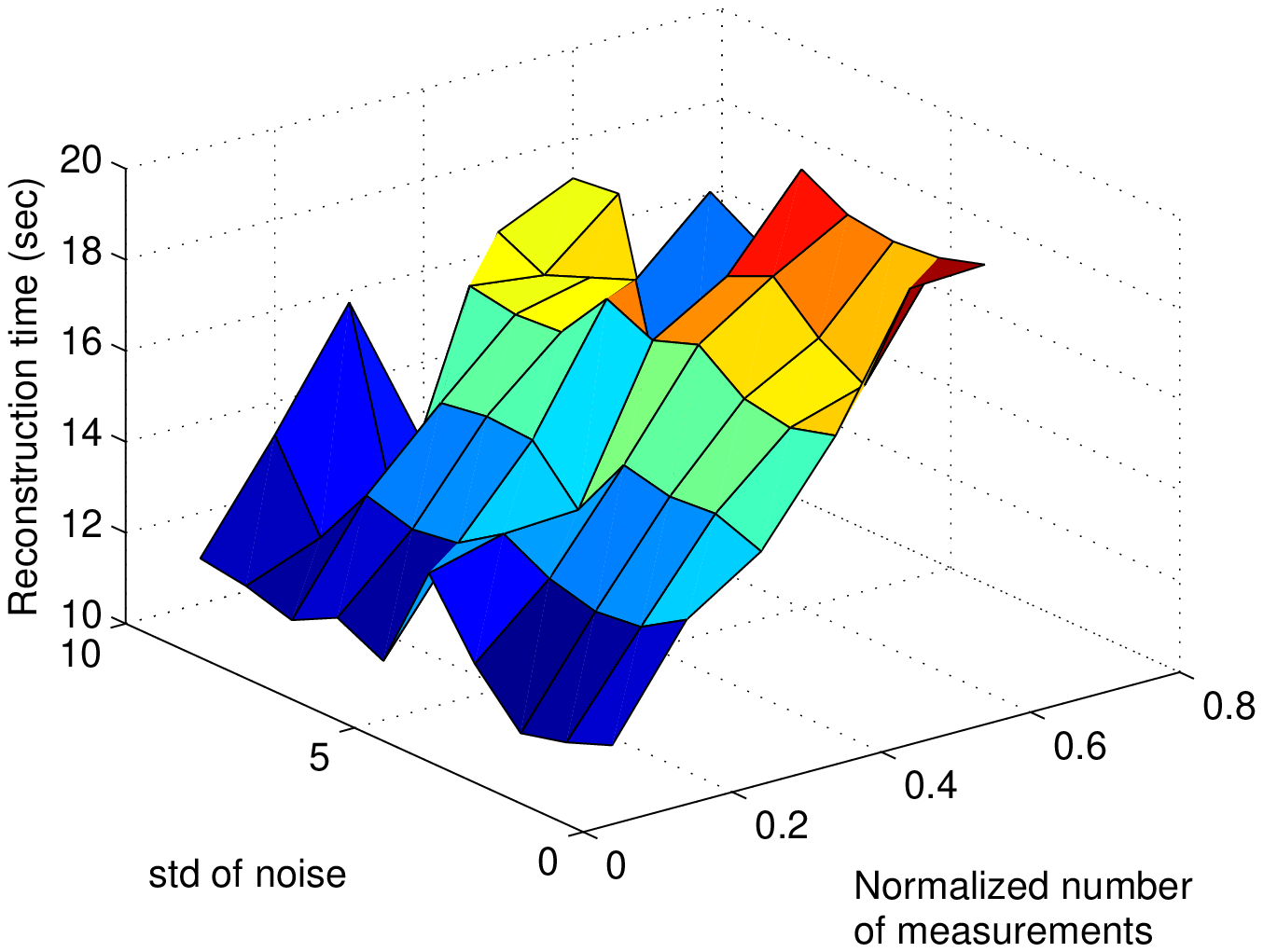}}
\subfigure[GTCS-P]{\label{UICTime}\includegraphics[width=0.32\linewidth]{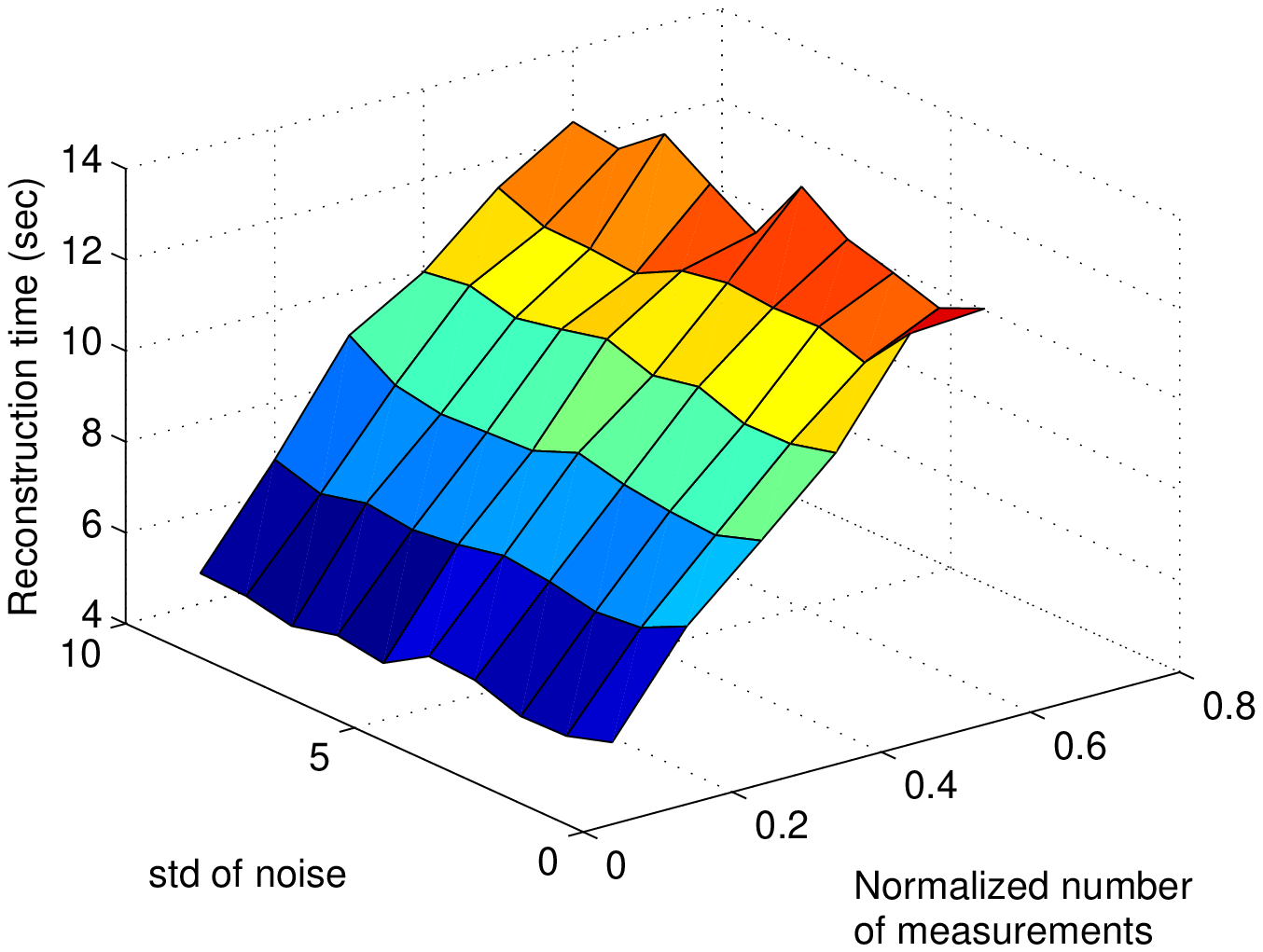}}
\subfigure[KCS]{\label{UICTime}\includegraphics[width=0.32\linewidth]{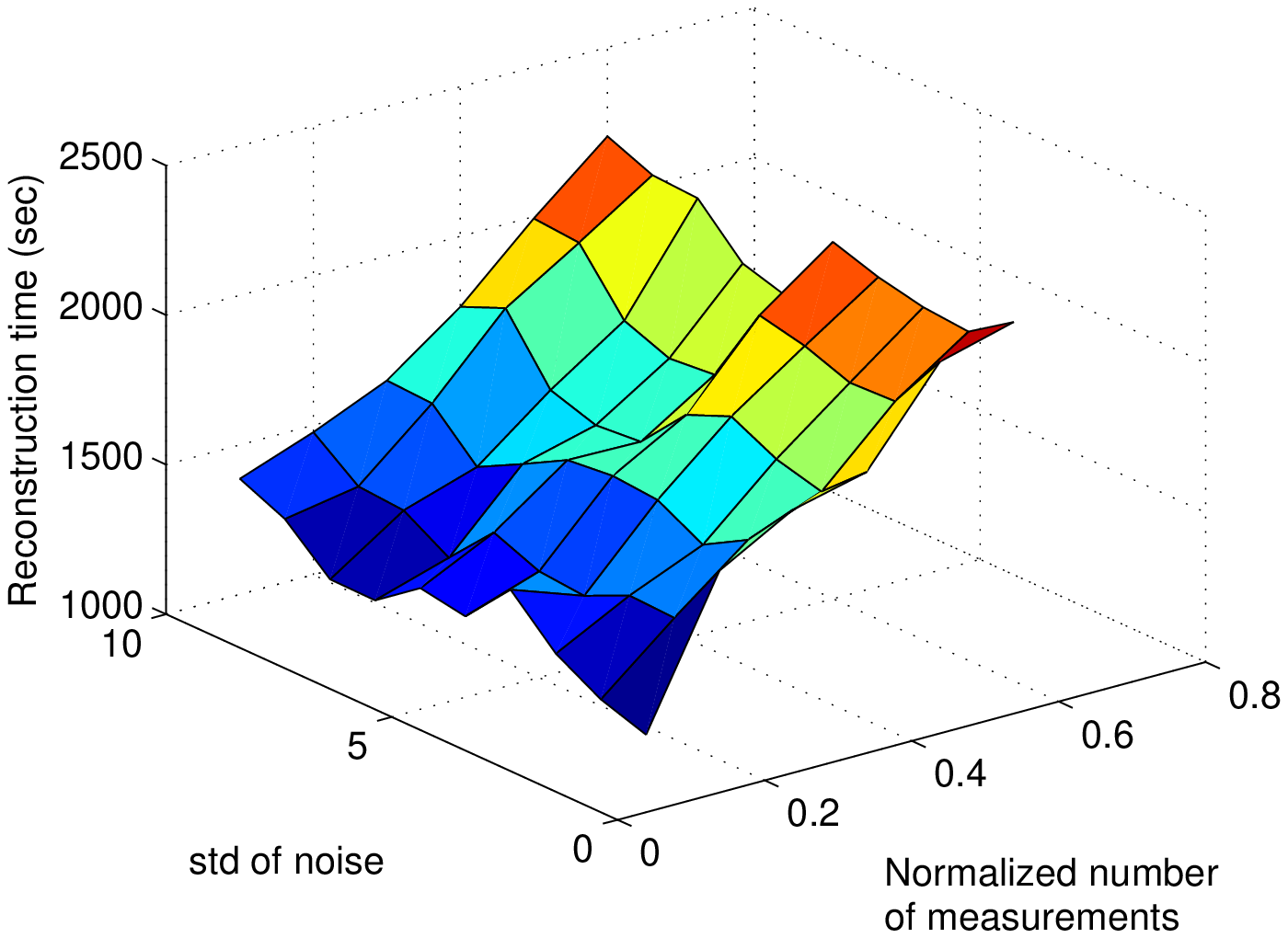}}
\caption{Execution time for the tested methods in the noisy recovery scenario.} \label{NoisyImageTime}
\end{center}
\end{figure}

\section{Compressed Sensing of Tensors}
\subsection{A Brief Introduction to Tensors}\label{multilinearAlgebra}
A tensor is a multidimensional array. The order of a tensor is the number of
modes. For instance, tensor $\mathcal{X}\in \mathbb{R}^{N_1\times \ldots\times
N_d}$ has order $d$ and the dimension of its $i^{th}$ mode (denoted mode
$i$) is $N_i$.

\begin{definition}
[Kronecker Product] {The Kronecker product between matrices $A\in
\mathbb{R}^{I\times J}$ and $B\in \mathbb{R}^{K\times L}$ is denoted by
$A\otimes B$. The result is the matrix of dimensions $(I\cdot K)\times (J\cdot L)$
defined by

$A\otimes B=\left(%
\begin{array}{cccc}
  a_{11}B & a_{12}B & \cdots & a_{1J}B \\
  a_{21}B & a_{22}B & \cdots & a_{2J}B \\
  \vdots & \vdots & \ddots & \vdots \\
  a_{I1}B & a_{I2}B & \cdots & a_{IJ}B \\
\end{array}%
\right)$}.
\end{definition}

\begin{definition}
[Outer Product and Tensor Product] {The operator $\circ$ denotes the tensor product between two vectors. In linear algebra, the outer product typically refers to the tensor product between two vectors, that is, $u\circ v=uv\trans$. In this chapter, the terms outer product and tensor product are equivalent. The Kronecker
product and the tensor product between two vectors are related by $u\circ v=u\otimes v\trans.$}
\end{definition}
\begin{definition}
[Mode-$i$ Product] {The mode-$i$ product of a tensor
$\mathcal{X}=[x_{\alpha_1,\ldots,\alpha_d}]\in \mathbb{R}^{N_1\times
\ldots\times N_d}$ and a matrix $U=[u_{j,\alpha_i}]\in \mathbb{R}^{J\times
N_i}$ is denoted by $\mathcal{X}\times _i U$ and is of size
$N_1\times\ldots\times N_{i-1}\times J \times N_{i+1}\times \ldots\times N_d$.
Element-wise, the mode-$i$ product can be written as $(\mathcal{X}\times _i
U)_{\alpha_1,\ldots,\alpha_{i-1},j,\alpha_{i+1},\ldots,\alpha_d} =
\sum_{\alpha_i=1}^{N_i}x_{\alpha_1,\ldots,\alpha_d}u_{j,\alpha_i}$.}
\end{definition}
\begin{definition}
[Mode-$i$ Fiber and Mode-$i$ Unfolding] {The mode-$i$ fiber of tensor
$\mathcal{X}=[x_{\alpha_1,\ldots,\alpha_d}]\in
\mathbb{R}^{N_1\times\ldots\times N_d}$ is the set of vectors obtained by fixing every index but $\alpha_i$. The mode-$i$ unfolding $X_{(i)}$ of $\mathcal{X}$ is the $N_i\times (N_1\cdot\ldots\cdot
N_{i-1}\cdot N_{i+1}\cdot\ldots\cdot N_d)$ matrix whose columns are the mode-$i$
fibers of $\mathcal{X}$. $\mathcal{Y} = \mathcal{X}\times_1 U_1\times \ldots\times_d U_d$ is equivalent
to $Y_{(i)} = U_iX_{(i)}(U_d\otimes \ldots\otimes U_{i+1}\otimes U_{i-1}\otimes
\ldots\otimes U_1)\trans$.}
\end{definition}

\begin{definition}
[Core Tucker Decomposition]\cite{Tuck1964}{ Let $\cX\in
\R^{N_1\times \ldots \times N_d}$ be a tensor with mode-$i$
unfolding $X_{(i)}\in \R^{N_i\times (N_1\cdot\ldots \cdot
N_{i-1}\cdot N_{i+1}\cdot\ldots \cdot N_d)}$ such that
$\rank(X_{(i)})=r_i$. Let $R_i(\cX)\subset\R^{N_i}$ denote the
column space of $X_{(i)}$, and $\c_{1,i},\ldots, \c_{r_i,i}$ be a
basis in $R_i(\cX)$. Then $\cX$ is an element of the subspace
$\V(\cX):=R_1(\cX)\circ\ldots\circ R_d(\cX) \subset \R^{N_1\times
\ldots \times N_d}$. Clearly, vectors $\c_{i_1,1}\circ\ldots\circ
\c_{i_d,d}$, where $i_j\in [r_j]$ and $j\in [d]$, form a basis of
$\V$.  The core Tucker decomposition of $\cX$ is
\begin{equation}\label{cortcukdec}
\cX=\sum_{i_j\in[r_j],j\in[d]} \xi_{i_1,\ldots,i_d}
\c_{i_1,1}\circ\ldots\circ \c_{i_d,d}
\end{equation}
for some decomposition coefficients $\xi_{i_1,\ldots,i_d}$,
$i_j\in [r_j]$ and $j\in [d]$.}\end{definition}

A special case of the core Tucker decomposition is the
higher-order singular value decomposition (HOSVD). Any tensor
$\cX\in \R^{N_1\times \ldots \times N_d}$ can be written as

\begin{equation}\label{HOSVD}
\cX=\cS\times_1 U_1\times\ldots\times_d U_d,
\end{equation}
where $U_i = [u_1\cdots u_{N_i}]$ is an orthonormal matrix for
$i\in [d]$, and $\mathcal{S} = \mathcal{X}\times_1
U_1\trans\times\ldots\times_d U_d\trans$ is called the core
tensor. For a more in-depth discussion on HOSVD, including the set
of properties the core tensor is required to satisfy, please refer
to \cite{HOSVD}.

$\cX$ can also be expressed in terms of weaker decompositions of
the form
 \begin{equation}\label{ranklikedec}
 \cX=\sum_{i=1}^K a_i^{(1)}\circ\ldots\circ a_i^{(d)}, \quad a_i^{(j)}\in R_j(\cX), j\in [d].
 \end{equation}
For instance, first decompose $X_{(1)}$ as
$X_{(1)}=\sum_{j=1}^{r_1} \c_{j,1} \g_{j,1}\trans$ (e.g.,\ via
SVD); then each $\g_{j,1}$ can be viewed as a tensor of order
$d-1$ $\in R_2(\cX)\circ \ldots\circ
R_d(\cX)\subset\mathbb{R}^{N_2\times \ldots\times N_d}$. Secondly,
unfold each $\g_{j,1}$ in mode $2$ to obtain ${\g_{j,1}}_{(2)}$
and decompose
${\g_{j,1}}_{(2)}=\sum_{l=1}^{r_2}d_{l,2,j}\f_{l,2,j}\trans, \quad
d_{l,2,j}\in R_2(\cX), f_{l,2,j}\in R_3(\cX)\circ\ldots\circ
R_d(\cX).$ By successively unfolding and decomposing each
remaining tensor mode, a decomposition of the form in
Eq.~\eqref{ranklikedec} is obtained.
 Note that if $\cX$ is $k$-sparse, then each vector in $R_i(\cX)$ is $k$-sparse and $r_i\leq k$ for $i\in [d]$. Hence, $K\le k^{d-1}$.
\begin{definition}
[CANDECOMP/PARAFAC
Decomposition]\cite{Kolda09tensordecompositions}{ For a tensor
$\mathcal{X}\in \mathbb{R}^{N_1\times \ldots\times N_d}$, the
CANDECOMP/PARAFAC (CP) decomposition is defined as
$\mathcal{X}\approx [\lambda;
A^{(1)},\ldots,A^{(d)}]\equiv\sum_{r=1}^R\lambda_r
a_r^{(1)}\circ\ldots\circ a_r^{(d)},$ where
$\lambda=[\lambda_1\ldots\lambda_R]\trans\in \mathbb{R}^R$ and
$A^{(i)}=[a_1^{(i)}\cdots a_R^{(i)}]\in \mathbb{R}^{N_i\times R}$
for $i\in [d].$}
\end{definition}

\subsection{Noiseless Recovery}
\subsubsection{Generalized Tensor Compressed Sensing - Serial Recovery (GTCS-S)} \label{sec:2}
The serial recovery method for compressed sensing of tensors in
the noiseless case is described by the following theorem.
\begin{theorem}\label{compsenstenrev}
Let $\cX=[x_{i_1,\ldots,i_d}]\in\R^{N_1\times \ldots \times N_d}$
be $k$-sparse. Let $U_i\in \R^{m_i\times N_i}$ and assume that
$U_i$ satisfies the NSP$_k$ property for $i\in [d]$. Define
\begin{equation}\label{deftenYtensor}
\cY=[y_{j_1,\ldots,j_d}]=\cX\times_1 U_1\times \ldots \times_d U_d
\in \R^{m_1\times\ldots\times m_d}.
\end{equation}
Then $\cX$ can be recovered uniquely as follows. Unfold $\cY$ in
mode $1$,
\begin{equation*}
Y_{(1)} = U_1X_{(1)}[\otimes_{k=d}^2 U_k]\trans\in\R^{m_1\times
(m_2\cdot\ldots\cdot m_d)}.
\end{equation*}
Let $y_1,\ldots,y_{m_2\cdot\ldots\cdot m_d}$ be the columns of
$Y_{(1)}$. Then $y_i=U_1z_i$, where each $z_i\in \R^{N_1}$ is
$k$-sparse. Recover each $z_i$ using Eq.~\eqref{l1minrec}. Let
$\cZ = \cX \times_2 U_2\times \ldots \times_d U_d \in
\R^{N_1\times m_2\times\ldots\times m_d}$, and let
$z_1,\ldots,z_{m_2\cdot\ldots\cdot m_d}$ denote its mode-$1$
fibers. Unfold $\cZ$ in mode 2,
\begin{equation*}
Z_{(2)} = U_2X_{(2)}[\otimes_{k=d}^3 U_k\otimes
I]\trans\in\R^{m_2\times (N_1\cdot m_3\cdot\ldots\cdot m_d)}.
\end{equation*}
Let $w_1,\ldots,w_{N_1\cdot m_3\cdot\ldots\cdot m_d}$ be the
columns of $Z_{(2)}$. Then $w_j=U_2v_j$, where each $v_j\in
\R^{N_2}$ is $k$-sparse. Recover each $v_j$ using
Eq.~\eqref{l1minrec}. $\cX$ can be reconstructed by successively
applying the above procedure to tensor modes $3,\ldots, d$.
\end{theorem}
\begin{proof}
The proof of this theorem is a straightforward generalization of
that of Theorem \ref{compsensmatS}.\qed
\end{proof}

Note that although Theorem \ref{compsenstenrev} requires $U_i$ to
satisfy the NSP$_k$ property for $i\in [d]$, such constraints can
be relaxed if each mode-$i$ fiber of $\mathcal{X}\times_{i+1}
U_{i+1}\times \ldots \times_d U_d$ is $k_i$-sparse for $i\in
[d-1]$, and each mode-$d$ fiber of $\mathcal{X}$ is $k_d$-sparse,
where $k_i\leq k$, for $i\in [d]$. In this case, it follows from
the proof of Theorem \ref{compsenstenrev} that $X$ can be
recovered as long as $U_i$ satisfies the NSP$_{k_i}$ property, for
$i\in [d]$.

\subsubsection{Generalized Tensor Compressed Sensing - Parallelizable Recovery
(GTCS-P)}\label{GTCSPNoiselessSection} The parallelizable recovery
method for compressed sensing of tensors in the noiseless case is
described by the following theorem.
\begin{theorem}[GTCS-P]\label{TheoremGTCSP}
Let $\cX=[x_{i_1,\ldots,i_d}]\in\R^{N_1\times \ldots \times N_d}$
be $k$-sparse. Let $U_i\in \R^{m_i\times N_i}$ and assume that
$U_i$ satisfies the NSP$_k$ property for $i \in [d]$. If $\cY$ is
given by Eq.~\eqref{deftenYtensor}, then $\cX$ can be recovered
uniquely as follows. Consider a decomposition of $\cY$ such that,
\begin{align}\label{ranklikedecY}
&\cY=\sum_{i=1}^K b_i^{(1)}\circ\ldots\circ b_i^{(d)}, \quad
b_i^{(j)}\in R_j(\cY)\subseteq U_jR_j(\cX), j\in [d].
\end{align}
Let $\hat w_i^{(j)}\in R_j(\cX)\subset \mathbb{R}^{N_j}$ be a
solution of
\begin{align}\label{defzstarP}
 &\hat w_i^{(j)}=\arg\min_{w_i^{(j)}}\|w_i^{(j)}\|_1\quad
\text{s.t.}\quad b_i^{(j)} = U_jw_i^{(j)}, \quad
 i\in [K], j\in [d].
\end{align}
Thus each $\hat w_i^{(j)}$ is unique and $k$-sparse. Then,
 \begin{equation}\label{ranklikedecYP}
 \cX=\sum_{i=1}^K w_i^{(1)}\circ\ldots\circ w_i^{(d)}, \quad w_i^{(j)}\in
 R_j(\cX), j\in [d].
 \end{equation}
\end{theorem}
\begin{proof}
Since $\cX$ is $k$-sparse, each vector in $R_j(\cX)$ is
$k$-sparse. If each $U_j$ satisfies the NSP$_k$ property, then
$w_i^{(j)}\in R_j(\cX)$ is unique and $k$-sparse. Define $\cZ$ as
\begin{equation}\label{ranklikedecZ}
\cZ=\sum_{i=1}^K w_i^{(1)}\circ\ldots\circ w_i^{(d)}, \quad
w_i^{(j)}\in R_j(\cX), j\in [d].
\end{equation}
Then
\begin{equation}\label{inductioninitial}
(\cX-\cZ)\times_1 U_1\times\ldots \times_d U_d=0.
\end{equation}

To show $\cZ=\cX$, assume a slightly more general scenario, where
each $R_j(\cX)\subseteq \V_j \subset\R^{N_j}$, such that each
nonzero vector in $\V_j$ is $k$-sparse. Then $R_j(\cY)\subseteq
U_jR_j(\cX)\subseteq U_j\V_j$ for $j\in[d]$. Assume to the
contrary that $\cX\neq\cZ$. This hypothesis can be disproven via
induction on mode $m$ as follows.

Suppose
\begin{equation}\label{inductionm}
(\cX-\cZ)\times_m U_m\times\ldots\times_d U_d=0.
\end{equation}

Unfold $\cX$ and $\cZ$ in mode $m$, then the column (row) spaces
of $X_{(m)}$ and $Z_{(m)}$ are contained in $\V_m$
($\hat\V_m:=\V_1\circ\ldots\circ
\V_{m-1}\circ\V_{m+1}\circ\ldots\circ \V_d$). Since $\cX\neq\cZ$,
$X_{(m)}-Z_{(m)}\ne 0$. Then $X_{(m)}-Z_{(m)}=\sum_{i=1}^p u_
iv_i\trans$, where $\rank (X_{(m)}-Z_{(m)})=p$, and
$u_1,\ldots,u_p\in\V_m, v_1,\ldots,v_p\in \hat \V_m$ are two sets
of linearly independent vectors.

Since $(\cX-\cZ)\times_m U_m\times\ldots\times_d U_d=0$,
\begin{align*}
0&=U_m(X_{(m)}-Z_{(m)})(U_d\otimes\ldots\otimes U_{m+1}\otimes I)\trans\\
&=U_m(X_{(m)}-Z_{(m)})\hat U_m\trans\\
&=\sum_{i=1}^p (U_mu_i)(\hat U_mv_i)\trans.
\end{align*}

Since $U_mu_1,\ldots,U_mu_p$ are linearly independent (see
Appendix for proof), it follows that $\hat U_mv_i=0$ for $i\in
[p]$. Therefore,
\[(X_{(m)}-Z_{(m)})\hat U_m\trans=(\sum_{i=1}^p u_iv_i\trans)\hat U_m\trans = \sum_{i=1}^p u_i(\hat U_mv_i)\trans =
0,\] which is equivalent to (in tensor form, after folding)
\begin{align}\label{inductionm1}
&(\cX-\cZ)\times_m I_m\times_{m+1} U_{m+1}\times\ldots
\times_d U_d\nonumber\\
&=(\cX-\cZ)\times_{m+1} U_{m+1}\times\ldots \times_d U_d=0,
\end{align}
where $I_m$ is the $N_m\times N_m$ identity matrix. Note that Eq.
\eqref{inductionm} leads to Eq. \eqref{inductionm1} upon replacing
$U_m$ with $I_m$. Similarly, when $m=1$, $U_1$ can be replaced
with $I_1$ in Eq.~\eqref{inductioninitial}. By successively
replacing $U_m$ with $I_m$ for $2\leq m\leq d$,
 \begin{align*}
&(\cX-\cZ)\times_1 U_1\times\ldots \times_d U_d\\
=&(\cX-\cZ)\times_1 I_1\times\ldots \times_d I_d\\
=&\cX-\cZ=0, \label{eqstepdm1}
\end{align*}
which contradicts the assumption that $\cX\neq\cZ$. Thus,
$\cX=\cZ$. This completes the proof. \qed
\end{proof}

Note that although Theorem \ref{TheoremGTCSP} requires $U_i$ to
satisfy the NSP$_k$ property for $i\in [d]$, such constraints can
be relaxed if all vectors $\in R_i(\cX)$ are $k_i$-sparse. In this
case, it follows from the proof of Theorem \ref{TheoremGTCSP} that
$X$ can be recovered as long as $U_i$ satisfies the NSP$_{k_i}$,
for $i\in [d]$.

As in the matrix case, the reconstruction stage of the recovery
process can be implemented in parallel for each tensor mode.

Note additionally that Theorem \ref{TheoremGTCSP} does not require
tensor rank decomposition, which is an NP-hard problem. Weaker
decompositions such as the one described by Eq. \ref{ranklikedec}
can be utilized.

The above described procedure allows exact recovery. In some
cases, recovery of a rank-$R$ approximation of $\cX$,
$\mathcal{\hat{X}}=\sum_{r=1}^Rw_r^{(1)}\circ \ldots\circ
w_r^{(d)}$, suffices. In such scenarios, $\cY$ in Eq.
\eqref{ranklikedecY} can be replaced by its rank-$R$
approximation, namely, $\cY=\sum_{r=1}^Rb_r^{(1)}\circ \ldots\circ
b_r^{(d)}$ (obtained e.g., \ by CP decomposition).

\subsubsection{Simulation Results}\label{noisySimuTensor}
Examples of data that is amendable to tensorial representation
include color and multi-spectral images and video. We use a
24-frame, $24\times 24$ pixel grayscale video to test the
performance of our algorithm (see Fig. \ref{OriginalVideo}). In
other words, the video data is represented as a $24\times 24\times
24$ tensor ($N=13824$). We use the three-dimensional DCT as the
sparsifying transform, and zero-out coefficients outside the
$6\times 6\times 6$ cube located on the front upper left corner of
the transformed tensor. As in the image case, let $m$ denote the
number of measurements along each tensor mode; we generate the
measurement matrices with entries drawn from a Gaussian
distribution with mean $0$ and standard deviation
$\sqrt{\frac{1}{m}}$. For simplicity, we set the number of
measurements for each tensor mode to be equal; that is, the
randomly constructed Gaussian matrix $U$ is of size $m\times 24$
for each mode. Therefore, the KCS measurement matrix $U\otimes
U\otimes U$ is of size $m^3\times 13824$, and the total number of
measurements is $m^3$. We refer to $\frac{m^3}{N}$ as the
normalized number of measurements. For GTCS-P, we employ the
weaker form of the core Tucker decomposition as described in
Section \ref{multilinearAlgebra}. Although the reconstruction
stage of GTCS-P is parallelizable, we recover each vector in
series. We examine the performance of KCS and GTCS-P by varying
the normalized number of measurements from 0.1 to 0.6 in steps of
0.1. Reconstruction accuracy is measured in terms of the average
PSNR across all frames between the recovered and the target video,
whereas computational complexity is measured in terms of the
$\log$ of the reconstruction time (see Fig. \ref{NoiselessVideo}).
\begin{figure}[htb]
\begin{center}
\includegraphics[width=0.49\linewidth]{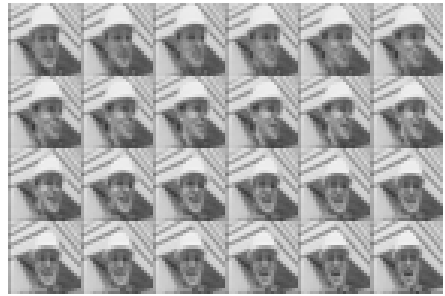}
\caption{The original 24 video frames.} \label{OriginalVideo}
\end{center}
\end{figure}

Note that in the tensor case, due to the serial nature of GTCS-S,
the reconstruction error propagates through the different stages
of the recovery process. Since exact reconstruction is rarely
achieved in practice, the equality constraint in the
$\ell_1$-minimization process described by Eq. \eqref{l1minrec}
becomes increasingly difficult to satisfy for the latter stages of
the reconstruction process. In this case, a relaxed recovery
procedure as described in Eq. \eqref{noisyrecoveryEqu} can be
employed. Since the relaxed constraint from Eq.
\eqref{noisyrecoveryEqu} results in what effectively amounts to
recovery in the presence of noise, we do not compare the
performance of GTCS-S with that of the other two methods.

\begin{figure}[htb]
\begin{center}
\subfigure[PSNR
comparison]{\label{SusiePSNR}\includegraphics[width=0.49\linewidth]{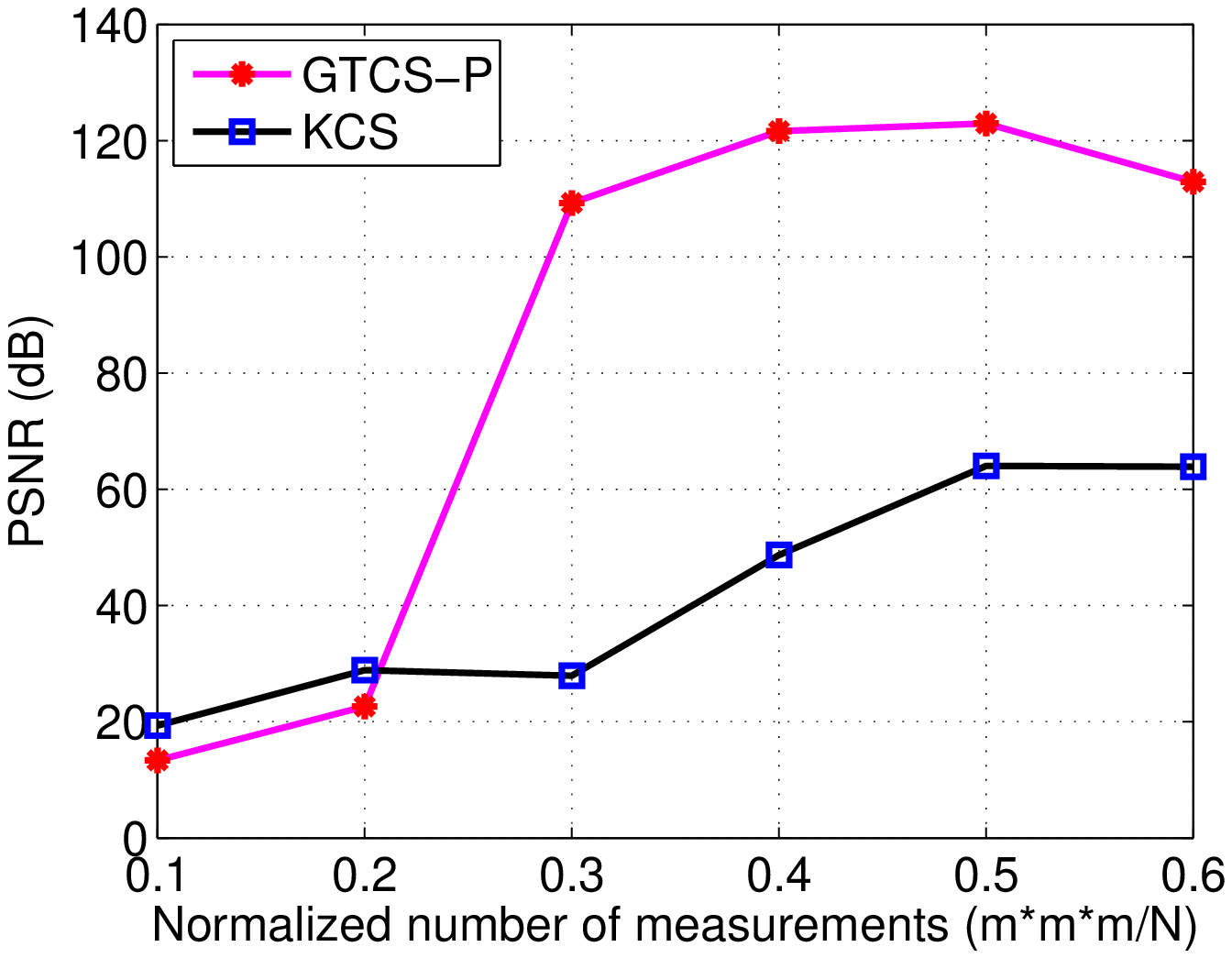}}
\subfigure[Recovery time
comparison]{\label{SusieTime}\includegraphics[width=0.49\linewidth]{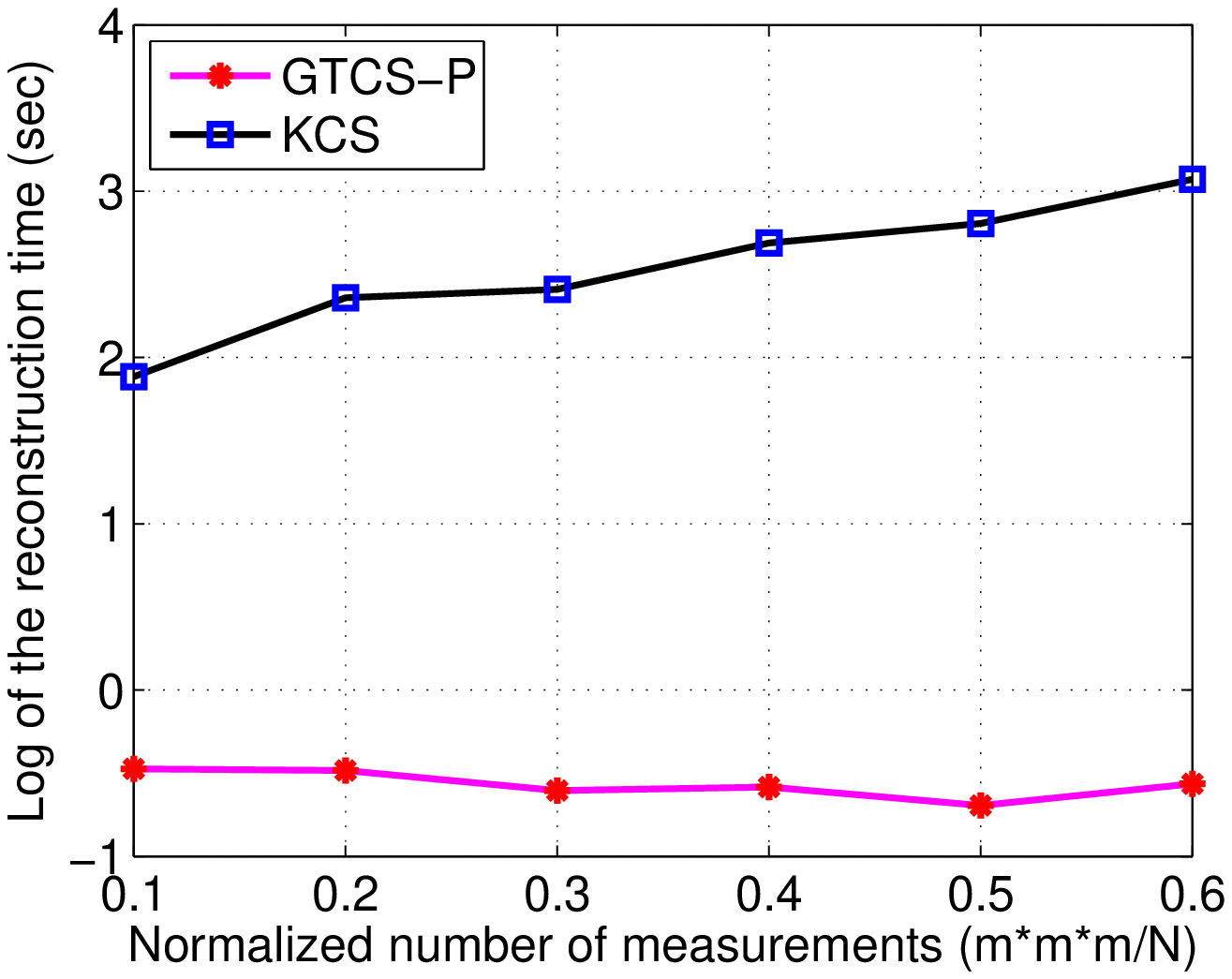}}
\caption{Performance comparison among the tested methods in terms
of PSNR and reconstruction time in the scenario of noiseless
recovery of the sparse video.} \label{NoiselessVideo}
\end{center}
\end{figure}

\subsection{Recovery in the Presence of Noise}
\subsubsection{Generalized Tensor Compressed Sensing - Serial Recovery (GTCS-S) in the Presence of Noise}
Let $\cX=[x_{i_1,\ldots,i_d}]\in\R^{N_1\times \ldots \times N_d}$
be $k$-sparse. Let $U_i\in \R^{m_i\times N_i}$ and assume that
$U_i$ satisfies the NSP$_k$ property for $i\in [d]$. Define
\begin{equation}\label{deftenY}
\cY=[y_{j_1,\ldots,j_d}]=\cX\times_1 U_1\times \ldots \times_d U_d
+\cE \in \R^{m_1\times\ldots\times m_d},
\end{equation}
where $\cE$ is the noise tensor and $\|\cE\|_F\leq \varepsilon$
for some real nonnegative number $\epsilon$. Although the norm of
the noise tensor is not equal across different stages of GTCS-S,
it is assumed that at any given stage, the entries of the error
tensor are independent and identically distributed. The upper
bound of the reconstruction error for GTCS-S recovery in the
presence of noise is derived next by induction on mode $k$.

When $k=1$, unfold $\cY$ in mode $1$ to obtain matrix
$Y_{(1)}\in\R^{m_1\times (m_2\cdot\ldots\cdot m_d)}$. Recover each
$z_i^{(1)}$ by
\begin{equation}
\hat z_i^{(1)}=\arg\min_{z_i^{(1)}} \|z_i^{(1)}\|_1\quad
\text{s.t.}\quad \|c_i(Y_{(1)}) - U_1z_i^{(1)}\|_2 \le
\frac{\varepsilon}{\sqrt{m_2\cdot\ldots\cdot m_d}}.
\end{equation}

Let $\hat Z^{(1)}=[\hat z_1^{(1)}\ldots \hat
z_{m_2\cdot\ldots\cdot m_d}^{(1)}]\in\R^{N_1\times
(m_2\cdot\ldots\cdot m_d)}$. According to Eq. \eqref{sigrecer1},
$\|\hat z_i^{(1)}-c_i(X_{(1)}[\otimes_{k=d}^2 U_k]\trans)\|_2\le
C_2\frac{\varepsilon}{\sqrt{m_2\cdot\ldots\cdot m_d}}$, and
$\|\hat Z^{(1)}-X_{(1)}[\otimes_{k=d}^2 U_k]\trans\|_F\le
C_2\varepsilon$. In tensor form, after folding, this is equivalent
to $\|\hat \cZ^{(1)}-\cX\times_2 U_2\times\ldots \times_d
U_d\|_F\le C_2\varepsilon$.

Assume when $k=n$, $\|\hat \cZ^{(n)}-\cX\times_{n+1}
U_{n+1}\times\ldots \times_d U_d\|_F\le C_2^n\varepsilon$ holds.
For $k=n+1$, unfold $\hat \cZ^{(n)}$ in mode $n+1$ to obtain $\hat
Z^{(n)}_{(n+1)}\in \R^{m_{n+1}\times (N_1\cdot\ldots\cdot
N_{n}\cdot m_{n+2}\cdot\ldots\cdot m_d)}$, and recover each
$z_i^{(n+1)}$ by

\begin{eqnarray}
&&\hat z_i^{(n+1)}=\arg\min_{z_i^{(n+1)}} \|z_i^{(n+1)}\|_1 \quad
\text{s.t.}\nonumber\\&&\|c_i(\hat Z^{(n)}_{(n+1)}) -
U_{n+1}z_i^{(n+1)}\|_2 \le
C_2^n\frac{\varepsilon}{\sqrt{N_1\cdot\ldots\cdot N_{n}\cdot
m_{n+2}\cdot\ldots\cdot m_d}}.
\end{eqnarray}

Let $\hat Z^{(n+1)}=[\hat z_1^{(n+1)}\ldots \hat
z_{N_1\cdot\ldots\cdot N_{n}\cdot m_{n+2}\cdot\ldots\cdot
m_d}^{(n+1)}]\in\R^{N_{n+1}\times (N_1\cdot\ldots\cdot N_{n}\cdot
m_{n+2}\cdot\ldots\cdot m_d)}$. Then $\|\hat
z_i^{(n+1)}-c_i(X_{(n+1)}[\otimes_{k=d}^{n+2} U_k]\trans)\|_2\le
C_2^{n+1}\frac{\varepsilon}{\sqrt{N_1\cdot\ldots\cdot N_{n}\cdot
m_{n+2}\cdot\ldots\cdot m_d}}$, and $\|\hat
Z^{(n+1)}-X_{(n+1)}[\otimes_{k=d}^{n+2} U_k]\trans\|_F\le
C_2^{n+1}\varepsilon$. Folding back to tensor form, $\|\hat
\cZ^{(n+1)}-\cX\times_{n+2} U_{n+2}\times\ldots \times_d
U_d\|_F\le C_2^{n+1}\varepsilon$.

When $k=d$, $\|\hat \cZ^{(d)}-\cX\|_F\le C_2^{d}\varepsilon$ by
induction on mode $k$.

\subsubsection{Generalized Tensor Compressed Sensing - Parallelizable Recovery (GTCS-P) in the Presence of Noise}
Let $\cX=[x_{i_1,\ldots,i_d}]\in\R^{N_1\times \ldots \times N_d}$
be $k$-sparse. Let $U_i\in \R^{m_i\times N_i}$ and assume that
$U_i$ satisfies the NSP$_k$ property for $i\in [d]$. Let $\cY$ be
defined as in Eq. \eqref{deftenY}. GTCS-P recovery in the presence
of noise operates as in the noiseless recovery case described in
Section \ref{GTCSPNoiselessSection}, except that $\hat w_i^{(j)}$
is recovered via
\begin{align}\label{defwnoisy}
 &\hat w_i^{(j)}=\arg\min_{w_i^{(j)}}\|w_i^{(j)}\|_1\quad
\text{s.t.}\quad \|U_jw_i^{(j)} - b_i^{(j)}\|_2 \leq
\frac{\epsilon}{2k}, \quad
 i\in [K], j\in [d].
 \end{align}
It follows from the proof of Theorem \ref{compsensmatPN} that the
recovery error of GTCS-P in the presence of noise between the
original tensor $\cX$ and the recovered tensor $\hat \cX$ is
bounded as follows:
\begin{equation*}
\|\hat \cX-\cX\|_F\le C_2^{d}\varepsilon.
\end{equation*}


\subsubsection{Simulation Results}
In this section, we use the same target video and experimental
settings used in Section \ref{noisySimuTensor}. We simulate the
noisy recovery scenario by modifying the observation tensor with
additive, zero-mean Gaussian noise having standard deviation
values ranging from 1 to 10 in steps of 1, and attempt to recover
the target video using Eq. \eqref{noisyrecoveryEqu}.  As before,
reconstruction performance is measured in terms of the average
PSNR across all frames between the recovered and the target video,
and in terms of $\log$ of reconstruction time, as illustrated in
Figs. \ref{NoisyVideoPSNR} and \ref{NoisyVideoTime}. Note that the
illustrated results correspond to the performance of the methods
for a given choice of upper bound on the $l_2$ norm in Eq.
\eqref{noisyrecoveryEqu}; the PSNR numbers can be further improved
by tightening this bound.

\begin{figure}[htb]
\begin{center}
\subfigure[GTCS-P]{\label{SusiePSNR}\includegraphics[width=0.49\linewidth]{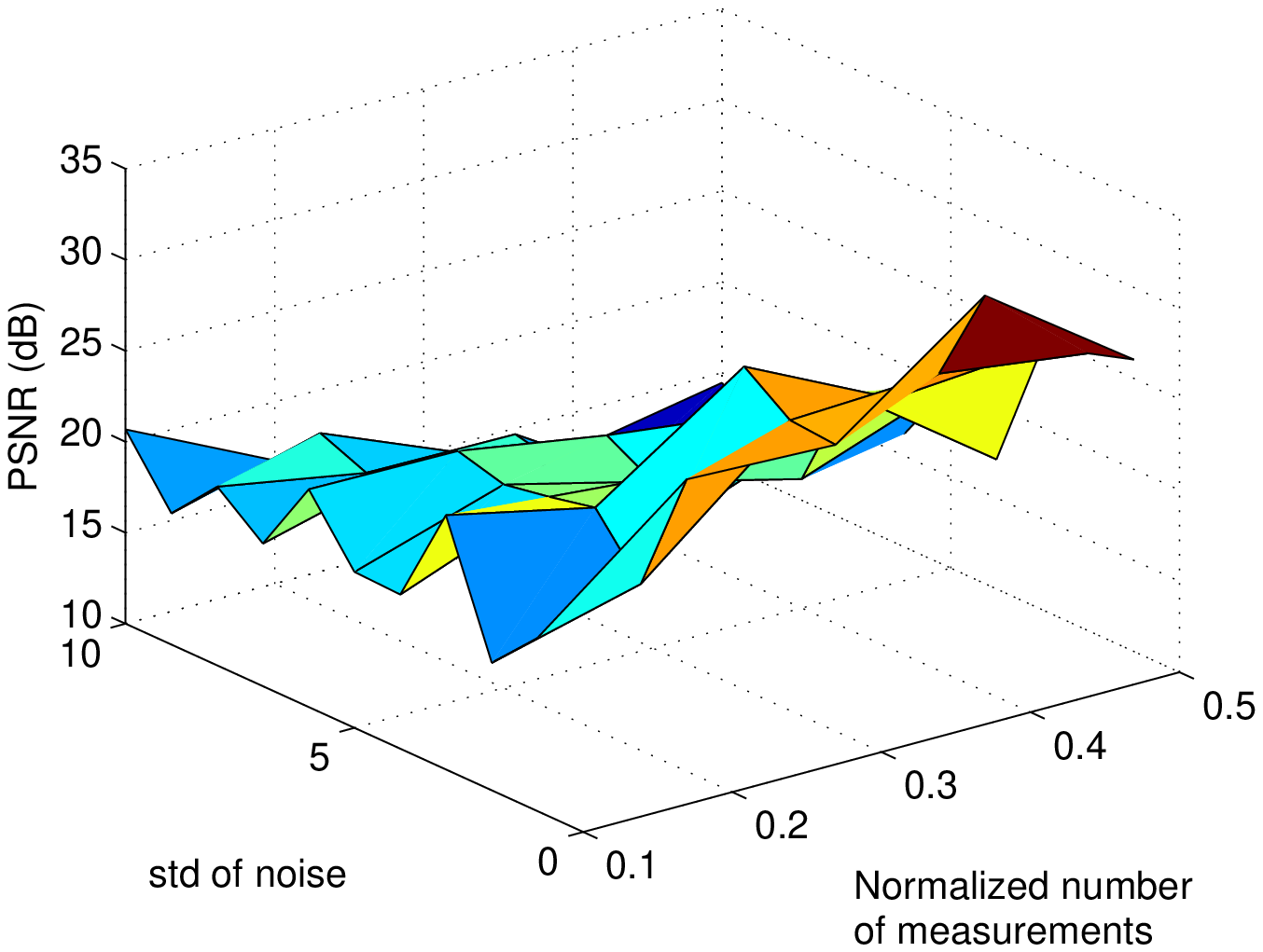}}
\subfigure[KCS]{\label{SusieTime}\includegraphics[width=0.49\linewidth]{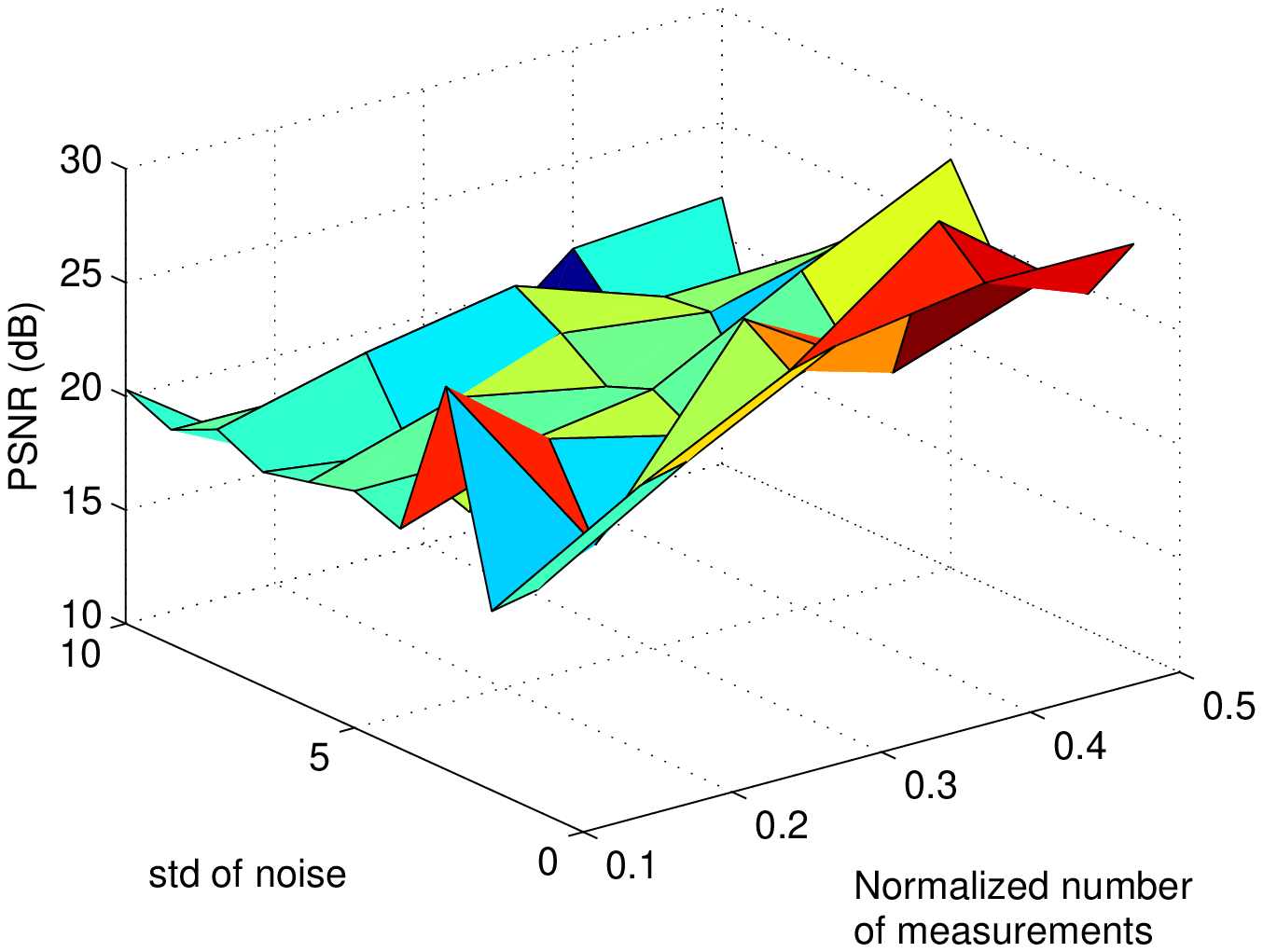}}
\caption{PSNR for the tested methods in the scenario of recovering
the sparse video in the presence of noise.} \label{NoisyVideoPSNR}
\end{center}
\end{figure}

\begin{figure}[htb]
\begin{center}
\subfigure[GTCS-P]{\label{SusiePSNR}\includegraphics[width=0.49\linewidth]{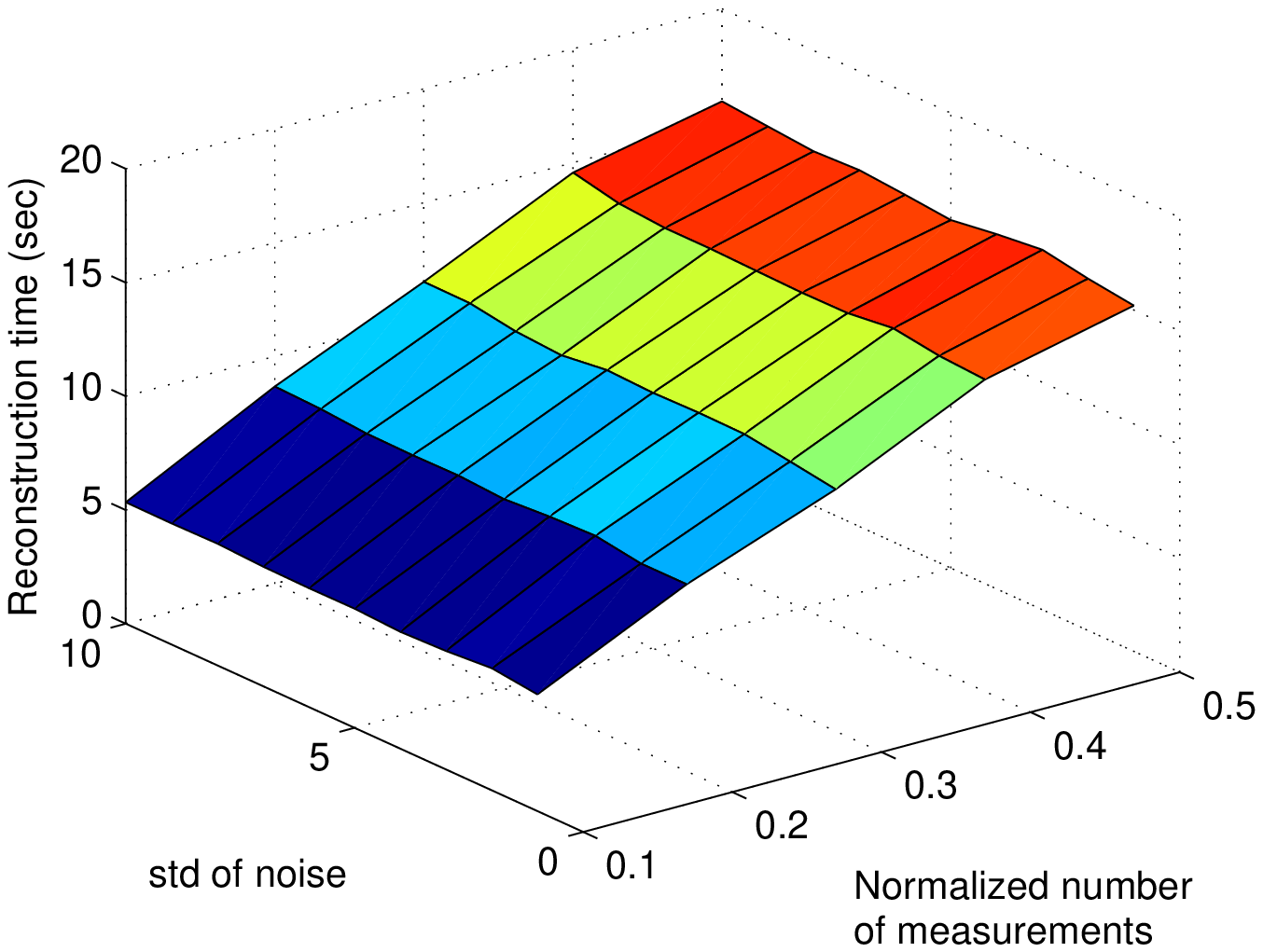}}
\subfigure[KCS]{\label{SusieTime}\includegraphics[width=0.49\linewidth]{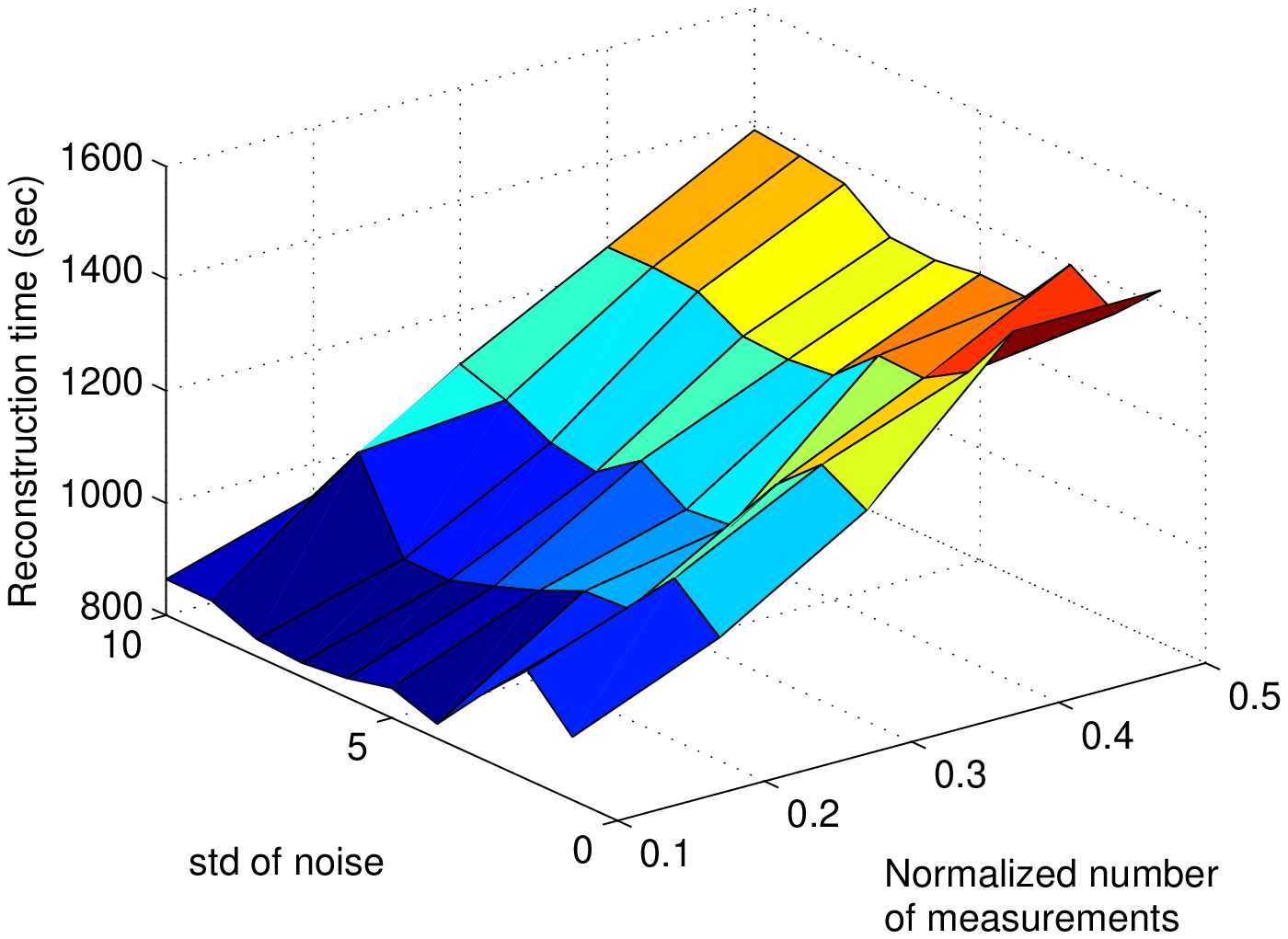}}
\caption{Execution time for the tested methods in the scenario of
recovering the sparse video in the presence of noise.}
\label{NoisyVideoTime}
\end{center}
\end{figure}

\subsection{Tensor Compressibility}
Let $\cX=[x_{i_1,\ldots,i_d}]\in\R^{N_1\times \ldots \times N_d}$.
Assume the entries of the measurement matrix are drawn from a
Gaussian or Bernoulli distribution as described above. For a given
level of reconstruction accuracy, the number of measurements for
$\cX$ required by GTCS should satisfy
\begin{equation}\label{comprsratten}
m\ge 2^dc^d \prod_{i\in[d]} \ln \frac{N_i}{k}.
\end{equation}
Suppose that $N_1=\ldots N_d=N^{\frac{1}{d}}$. Then
\begin{equation}\label{comprsratten1}
m\ge 2^dc^d  (\ln \frac{N^{\frac{1}{d}}}{k})^d=2^dc^d
(\frac{1}{d}\ln N-\ln k)^d.
\end{equation}
On the other hand, the number of measurements required by KCS
should satisfy
\begin{equation}\label{comprsratten2}
m\ge 2c \ln \frac{N}{k}.
\end{equation}

Note that the lower bound in Eq. \eqref{comprsratten2} is
indicative of a better compression ratio relative to that in Eq.
\eqref{comprsratten1}. In fact, this phenomenon has been observed in
simulations (see Ref.~\cite{FLS13}), which indicate that KCS
reconstructs the data with better compression ratios than GTCS.

\section{Conclusion}\label{conclusion}
In applications involving color images, video sequences, and
multi-sensor networks, the data is intrinsically of high-order,
and thus more suitably represented in tensorial form. Standard
applications of CS to higher-order data typically involve
representation of the data as long vectors that are in turn
measured using large sampling matrices, thus imposing a huge
computational and memory burden. As a result, extensions of CS
theory to multidimensional signals have become an emerging topic.
Existing methods include Kronecker compressed sensing (KCS) for
sparse tensors and multi-way compressed sensing (MWCS) for sparse
and low-rank tensors. KCS utilizes Kronecker product matrices as
the sparsifying bases and to represent the measurement protocols
used in distributed settings. However, due to the requirement to
vectorize multidimensional signals, the recovery procedure is
rather time consuming and not applicable in practice. Although
MWCS achieves more efficient reconstruction by fitting a low-rank
model in the compressed domain, followed by per-mode
decompression, its performance relies highly on the quality of the
tensor rank estimation results, the estimation being an NP-hard
problem. We introduced the Generalized Tensor Compressed Sensing
(GTCS)--a unified framework for compressed sensing of higher-order
tensors which preserves the intrinsic structure of tensorial data
with reduced computational complexity at reconstruction. We
demonstrated that GTCS offers an efficient means for
representation of multidimensional data by providing simultaneous
acquisition and compression from all tensor modes. We introduced
two reconstruction procedures, a serial method (GTCS-S) and a
parallelizable method (GTCS-P), both capable of recovering a
tensor based on noiseless and noisy observations, and compared the
performance of the proposed methods with Kronecker compressed
sensing (KCS) and multi-way compressed sensing (MWCS). As shown,
GTCS outperforms KCS and MWCS in terms of both reconstruction
accuracy (within a range of compression ratios) and processing
speed. The major disadvantage of our methods (and of MWCS as
well), is that the achieved compression ratios may be worse than
those offered by KCS. GTCS is advantageous relative to
vectorization-based compressed sensing methods such as KCS because
the corresponding recovery problems are in terms of a multiple
small measurement matrices $U_i$'s, instead of a single, large
measurement matrix $A$, which results in greatly reduced
complexity. In addition, GTCS-P does not rely on tensor rank
estimation, which considerably reduces the computational
complexity while improving the reconstruction accuracy in
comparison with other tensorial decomposition-based method such as
MWCS.

\section*{Appendix}
\addcontentsline{toc}{section}{Appendix}

Let $X=[x_{ij}]\in\R^{N_1\times N_2}$ be $k$-sparse. Let $U_i\in
\R^{m_i\times N_i}$, and assume that $U_i$ satisfies the NSP$_k$
property for $i\in [2]$.  Define $Y$ as
\begin{equation}
Y=[y_{pq}]=U_1 X U_2\trans \in \R^{m_1\times m_2}.
\end{equation}
Given a rank decomposition of $X$, $X=\sum_{i=1}^r z_i u_i\trans$,
where $\rank(X) = r$, $Y$ can be expressed as
\begin{equation}\label{appendixDecomp}
Y=\sum_{i=1}^r (U_1z_i)(U_2u_i)\trans,
\end{equation}
which is also a rank-$r$ decomposition of $Y$, where
$U_1\z_1,\ldots,U_1\z_r$ and $U_2\u_1,\ldots,U_2\u_r$ are two sets
of linearly independent vectors.
\begin{proof}
Since $X$ is $k$-sparse, $\rank (Y)\le \rank (X)\le k$.
Furthermore, both $R(X)$, the column space of $X$, and $R
(X\trans)$ are vector subspaces whose elements are $k$-sparse.
Note that $\z_i\in R (X), \u_i\in R(X\trans)$. Since $U_1$ and
$U_2$ satisfy the NSP$_k$ property, then $\dim (U_1 R(X))= \dim
(U_2R (X\trans))=\rank (X)$. Hence the decomposition of $Y$ in Eq.
\eqref{appendixDecomp} is a rank-$r$ decomposition of $Y$, which
implies that $U_1\z_1,\ldots,U_1\z_r$ and
$U_2\u_1,\ldots,U_2\u_r$ are two sets of linearly independent
vectors. This completes the proof. \qed
\end{proof}
%
%

\begin{thebibliography}{99.}%
%
%
\bibitem{CS1} Candes, E. J., Romberg, J. K., Tao, T.: Robust Uncertainty Principles: Exact Signal Reconstruction from Highly Incomplete Frequency Information. Information Theory, IEEE Transactions on , vol.52, no.2, pp.489-509, Feb.
2006
\bibitem{CS2} Donoho, D. L.: Compressed Sensing. Information Theory, IEEE Transactions on, vol.52, no.4, pp.1289-1306, Apr. 2006
\bibitem{NSP} Cohen, A., Dahmen, W., Devore, R.: Compressed
Sensing and Best k-Term Approximation. J. Amer. Math. Soc., vol.
22, no.1, pp.211-231, 2009
\bibitem{RIP}  Candes, E. J.: The Restricted Isometry Property and
its Implications for Compressed Sensing. Comptes Rendus Math.,
vol.346, nos.9-10, pp.589-592, 2008
\bibitem{HOSVD} De Lathauwer, L., De Moor, B., Vandewalle, J.: A
Multilinear Singular Value Decomposition, SIAM J. Matrix Anal.
Appl., vol.21, pp.1253-1278, 2000
\bibitem{KCS} Duarte, M.F., Baraniuk, R.G.: Kronecker Compressive Sensing. Image Processing, IEEE Transactions on , vol.21, no.2, pp.494-504, Feb. 2012
\bibitem{FLS13} Friedland, S., Li, Q., Schonfeld, D.: Compressive Sensing of Sparse Tensor, arXiv:1305.5777
\bibitem{MWCS} Sidiropoulos, N.D., Kyrillidis, A.: Multi-Way Compressed Sensing for Sparse Low-Rank Tensors, Signal Processing Letters, IEEE
,vol.19, no.11, pp.757-760, Nov. 201
\bibitem{SVD} Golub, G. H., Charles, F. V. L.:
 Matrix Computations, fourth edition
\bibitem{l1} Candes, E. J., Romberg, J. K.: The $l_1$ magic
 toolbox, available online: http://www.l1-magic.org
\bibitem{C2} Candes,E. J., Romberg, J. K., Tao, T.: Stable Signal Recovery from Incomplete and Inaccurate Measurements, Communications on Pure and Applied Mathematics, Vol. LIX, pp. 1207-1223, 2006
\bibitem{Tuck1964} Tucker, L. R: The Extension of Factor Analysis to Three-Dimensional
Matrices, Contributions to mathematical psychology, 1964
\bibitem{Kolda09tensordecompositions} Kolda, T. G., Bader, B. W.: Tensor Decompositions and
Applications, SIAM REVIEW, vol. 51, no.3, pp. 455-500, 2009
\bibitem{Rauhut_compressivesensing} Rauhut, Holger: Compressive Sensing and Structured Random Matrices, Radon Series Comp. Appl. Math XX, 1¨C94
\bibitem{Complexity} Donoho, D. L., Tsaig, Y.: Extensions of Compressed
Sensing, Signal Processing, vol. 86, no.3, pp.533-548, Mar. 2006
\bibitem{Fast} Donoho, D.L., Tsaig, Y.: Fast Solution of $l_1$-Norm Minimization Problems When the Solution May Be Sparse, Information Theory, IEEE Transactions on, vol.54, no.11, pp.4789-4812, Nov. 2008






%
\end{thebibliography}
%
\biblstarthook{}

\end{document}